\documentclass[reqno,12pt]{amsart}

\usepackage[a4paper, 
            left=1.in,
            right=1in,
            top=1in,
            bottom=1in, 
            footskip=.25in]{geometry}

\usepackage[utf8]{inputenc}
\usepackage{microtype}
\usepackage{graphicx}
\usepackage{amssymb} 
\usepackage[colorlinks = true,
            linkcolor = blue,
            urlcolor  = blue,
            citecolor = blue,
            anchorcolor = blue]{hyperref}

\usepackage{enumerate}
\usepackage{enumitem}
\usepackage{thmtools} 
\usepackage{thm-restate}

\usepackage{microtype}

\usepackage{xcolor}

\theoremstyle{plain}
\newtheorem{thm}{}[section]
\newtheorem{lemma}[thm]{Lemma}

\newtheorem{theorem}[thm]{Theorem}

\theoremstyle{remark}
\newtheorem{observation}[thm]{Observation}

\newtheorem{definition}[thm]{Definition}  
  
 \newtheorem{example}[thm]{Example}

\usepackage{centernot}  
\usepackage{xspace}

 \usepackage{stmaryrd}   

\usepackage[scr=boondoxo]{mathalpha}  
\newcommand{\coloneqq}{\mathrel{\mathop:}\mathrel{\mkern-1.2mu}=} 

\newcommand{\ass}[1]{\ensuremath{\llbracket#1\rrbracket}\xspace}

\newcommand{\struct}[1]{\mathfrak{#1}}     
\newcommand{\CSP}{\ensuremath{\mathrm{CSP}}\xspace}   
   
\newcommand{\PO}{\ensuremath{\mathrm{P}}\xspace}  
\newcommand{\PH}{\ensuremath{\mathrm{PH}}\xspace}  
\newcommand{\QCSP}{\ensuremath{\mathrm{QCSP}}\xspace}

\newcommand{\fm}{\ensuremath{\mathrm{fm}}\xspace}

\newcommand{\temp}[1]{#1}

\usepackage{caption}
\AtBeginDocument{\fontsize{12pt}{15pt}\selectfont}

\begin{document}

\title[]{The Polynomial Hierarchy and $\omega$-categorical CSPs}

\author{Santiago Guzm\'{a}n-Pro}  
\author{Jakub Rydval}

\address{Technische Universit\"{a}t Dresden, Dresden, Germany}
\email{santiago.guzman\_pro@tu-dresden.de}

\address{Technische Universit\"{a}t Wien, Vienna, Austria}
\email{jakub.rydval@tuwien.ac.at}

 \begin{abstract}   
In 2008, Bodirsky and Grohe showed that for every 
$\Pi_n^{\mathrm{P}}$-level of the Polynomial Hierarchy (PH) there are $\omega$-categorical Constraint Satisfaction Problems (CSPs) complete for this level.
We show that, in fact, there are $\omega$-categorical CSPs complete for any level of the PH.  
To this end, we use a recent result of Bodirsky, Kn\"{a}uer, and Rudolph for
constructing $\omega$-categorical CSPs from sentences of Monadic Second-Order logic (MSO) with certain preservation properties. 
As a secondary contribution, we develop a new tool for producing MSO sentences satisfying said preservation properties.
\end{abstract}  

 \thanks{\emph{Santiago Guzm\'{a}n-Pro}: This research is funded by the European Union (ERC, POCOCOP, 101071674). Views and opinions expressed are however those of the author(s) only and do not necessarily reflect those of the European Union or the European Research Council Executive Agency. Neither the European Union nor the granting authority can be held responsible for them.
 \\
 \emph{Jakub Rydval}: This research was funded in whole or in part by the Austrian Science Fund (FWF) [ESP 1571225]. For the purpose of Open Access, the authors have applied a CC BY public copyright licence to any Author Accepted Manuscript (AAM) version arising from this submission.}
%

\maketitle


\section{Introduction}\label{section:introduction}  

The \emph{Constraint Satisfaction Problem} (CSP) of a  structure $\struct{B}$ in a finite relational signature $\tau$, denoted by $\CSP(\struct{B})$, is the computational problem of deciding whether a given
finite $\tau$-structure $\struct{A}$ homomorphically maps to $\struct{B}$. 
The general CSP framework is remarkably expressive; in fact, it contains all decision problems up to polynomial-time Turing equivalence~\cite{bodirsky2008non}.  
For this reason, CSPs are typically only studied under additional structural restrictions on $\struct{B}$.
One natural  restriction is requiring the domain of the \emph{template} structure $\struct{B}$ to be finite. 
The class of finite-domain CSPs famously constitutes a large fragment of \textup{\textrm{NP}} that admits a dichotomy between \textup{\textrm{P}} and \textup{\textrm{NP}}-completeness; this was proved in 2017 independently by Bulatov~\cite{bulatov2017dichotomy} and Zhuk~\cite{zhuk2017proof,zhuk2020proof}. 

Besides its complexity-theoretic significance (in particular, of confirming the dichotomy conjecture of Feder and Vardi~\cite{federvardi1998}), the result of Bulatov and Zhuk has the additional appeal that the border for tractability can be described by a neat
algebraic condition for the \emph{polymorphisms} of $\struct{B}$~\cite{barto2018wonderland,Siggers_2010}.
A polymorphism of a relational structure $\struct{B}$ is simply a homomorphism from a finite power of $\struct{B}$ into $\struct{B}$ itself.

The  success of the \emph{algebraic approach} in the finite-domain setting naturally raises the question of how far it extends to the infinite-domain case.
Although infinite-domain CSPs in general can attain arbitrary complexities, there are interesting subclasses whose template
cannot be captured by any finite-domain CSP, yet whose complexity is still determined by the polymorphisms of the template.
%
One of the most celebrated examples is the family of \emph{temporal} CSPs, i.e., CSPs of first-order reducts of $(\mathbb Q; <)$. Similar to finite-domain CSPs, temporal CSPs also enjoy a P vs.\ NP-complete dichotomy with a clean algebraic counterpart~\cite{ComplOfTempCSPs}.
Based on strong empirical evidence, Bodirsky and Pinsker~\cite{bodirsky2021projective} conjectured
that a P vs.\ NP-complete dichotomy characterized in terms of polymorphisms extends to all CSPs of \emph{reducts of finitely bounded homogeneous} structures, forming a computationally well-behaved subclass of \emph{$\omega$-categorical} structures~\cite{barto2019equations}, and providing an infinite-domain analogue of the dichotomy theorem of Bulatov and Zhuk.

The notion of $\omega$-categoricity  was popularized by Bodirsky and Ne\v{s}et\v{r}il~\cite{bodirsky2006constraint} as a sufficient restriction ensuring that primitive positive definability of relations is determined by polymorphisms.
A countable structure $\struct{A}$ is $\omega$-categorical if, for every positive integer $k$, the number of $k$-ary relations first-order definable in $\struct{A}$ is finite~\cite{hodges_book}; with $(\mathbb Q; <)$ being a prototypical example.
Bodirsky and Ne\v{s}et\v{r}il showed that if two countable $\omega$-categorical structures share the same set of polymorphisms, then they also share the same set of primitive positive definable relations~\cite{bodirsky2006constraint}.
Since expanding a CSP template by a primitive positive definable relation does not impact the complexity of the CSP (up to \textsc{logspace}-reductions), we say that the complexity of an $\omega$-categorical CSP is determined by the polymorphisms of its template.
Clearly, all finite structures are $\omega$-categorical, but this class also contains numerous interesting infinite structures, and
most of the research on the \emph{algebraic methods} for
infinite-domain CSPs has focused on this setting~\cite{10175732,barto2019equations,barto2018wonderland,barto_pinsker_journal,11186274,mottet2024smooth}.

We thus believe 
it is a central question 
to better understand which complexity classes 
appear in the complexity landscape of $\omega$-categorical CSPs, and more specifically, 
which computational problems can be represented as $\omega$-categorical CSPs.

Bodirsky and Grohe~\cite{bodirsky2008non} in fact already considered the first question in 2008.
From Theorem~2 in their paper  it follows that, if $\mathcal{C}$ is any complexity class for which there exist $\mathrm{coNP}^{\mathcal{C}}$-complete problems, then there exists an $\omega$-categorical structure $\struct{B}$ with a finite relational signature such that $\CSP(\struct{B})$ is $\mathrm{coNP}^{\mathcal{C}}$-complete.
In particular, there are $\omega$-categorical CSPs complete for the  $\Pi_n^{\PO}$-levels of the \emph{Polynomial Hierarchy} (\textrm{PH}).
This result later received a universal-algebraic upgrade~\cite[Theorem~1.8]{gillibert2022symmetries}, showing that the algebraic conditions  for polymorphisms (\emph{non-trivial identities}) that are expected to witness tractability of $\omega$-categorical CSPs in \textrm{NP} can actually be satisfied by the polymorphisms of $\omega$-categorical structures whose CSPs are complete for any complexity class lying above \textrm{coNP}, including \textrm{coNP} itself. 
Furthermore, the construction of $\mathrm{coNP}^{\mathcal{C}}$-complete $\omega$-categorical CSPs from~\cite{bodirsky2008non} also allows one to replicate Ladner’s method of ``blowing holes
into SAT''. 
This was used by Bodirsky and Grohe~\cite[Theorem~4]{bodirsky2008non} to show the existence of \textrm{coNP}-intermediate $\omega$-categorical CSPs (unless $\textrm{P}=\textrm{coNP}$). 
In other words, there is no dichotomy between \textrm{P} and \textrm{coNP}-completeness for $\omega$-categorical CSPs in \textrm{coNP}. 

All results mentioned in the previous paragraph follow a similar strategy.
First, given a formal language $\mathcal{L}$, its complement is turned into a set $\mathcal{F}$ of structures with complete
Gaifman graphs.\footnote{Two domain elements are adjacent in the Gaifman graph if they appear jointly in a relational tuple.}
Next, a class of finite structures $\mathcal{K}$ is defined by forbidding homomorphisms from all structures in $\mathcal{F}$.
Any such class $\mathcal{K}$ possesses two important structural properties: preservation under unions
(equivalently, the \emph{free amalgamation property}~\cite{bodirsky2021complexity}) and preservation under inverse homomorphisms.
 Via Fra\"{i}ss\'{e}'s theorem~\cite{hodges_book}, this is enough to ensure that $\mathcal K$ is the CSP of an $\omega$-categorical structure $\struct B$. 
 Solving the CSP of this structure $\struct B$ amounts to verifying  whether every complete substructure of an instance $\struct A$ is not an encoding of a word from the complement of $\mathcal{L}$. 
This technique therefore naturally yields problems complete for any complexity class of the
form $\mathrm{coNP}^{\mathcal{C}}$, while it is unlikely to produce completeness $\CSP(\struct B)$
for classes of the form $\mathrm{NP}^{\mathcal C}$. Indeed, to the best of our knowledge, no 
$\Sigma^\PO_{n}$-complete $\omega$-categorical CSP has appeared in the literature for any $n\geq 2$ 
In the present article, we close this gap. 
To this end, we follow an alternative construction of $\omega$-categorical CSPs
which, in contrast with the Bodirsky--Grohe construction, is not inherently tied to the
$\mathrm{coNP}^{\mathcal{C}}$ form, and we use it to exhibit examples of $\Sigma^\PO_{n}$-complete
$\omega$-categorical CSPs for $n\geq 2$.

\subsection{Contribution}

In 2020, Bodirsky, Kn\"{a}uer, and Rudolph~\cite{bodirsky2020datalog}  identified a new general sufficient condition for the existence of $\omega$-categorical CSPs.
Intuitively, this condition relaxes the preservation under unions to disjoint unions, compensated by an additional logical definability requirement.

\begin{theorem}[cf.~Corollary~17 in~\cite{bodirsky2020datalog}] \label{theorem:omega_cat_CSPs}
Let $\mathcal{K}$ be a class of finite structures that is preserved under inverse homomorphisms and disjoint unions and definable in monadic second-order logic.
Then there exists a countable $\omega$-categorical structure $\struct{D}$ such that $\mathcal{K} = \CSP(\struct{D})$.   
\end{theorem}

We employ Theorem~\ref{theorem:omega_cat_CSPs} as a base construction method for showing that there exist $\omega$-categorical CSPs complete for any level of the \textrm{PH}. 
Our key technical contribution  concerns preservation properties of SO sentences (Section~\ref{section:syntactic}).
Firstly, we pin-point a syntactic condition, which we call \emph{$\forall$-restriction}, that characterizes SO sentences preserved under substructures. 
Combining \emph{$\forall$-restriction} with the standard notion of \emph{negativity}, we characterize SO sentences preserved under inverse
injective homomorphisms (Section~\ref{section:dual}).  
Secondly, we identify a syntactic transformation of SO sentences\textemdash which we call the ``CSP hammer''\textemdash that turns 
$\forall$-restricted negative MSO sentences into MSO sentences that are further preserved under inverse homomorphisms and disjoint unions (Section~\ref{subsect:hammer}).
This yields a general method for constructing $\omega$-categorical CSPs, which is fundamental for our applications in Sections~\ref{sect:odd_part} and~\ref{sect:even_part}.
There, we construct specific MSO sentences $\Phi_n$ with such syntactic properties, and suitable for the CSP hammer.
These sentences $\Phi_n$ are tailored so that their adjustment using the CSP hammer produces $\Sigma_n^{\mathrm{P}}$-complete $\omega$-categorical CSPs. 
On the way, we establish a previously unknown connection between \emph{bounded alternation Quantified CSPs} (QCSPs) and $\omega$-categorical CSPs---QCSPs will be shortly introduced.

\begin{theorem} \label{thm:qcsp_to_csp}
For every finite structure $\struct{B}$ and every quantifier-prefix
$\mathrm{Q}_1\cdots \mathrm{Q}_n$ with at least one universal quantifier,
there exists an $\omega$-categorical structure $\struct{D}$ such that:
\begin{itemize}
    \item $\CSP(\struct{D})$ is expressible in $\mathrm{Q}_1\cdots \mathrm{Q}_n$MSO;
    \item $\mathrm{Q}_1\cdots \mathrm{Q}_n$$\QCSP(\struct{B})$ reduces in polynomial-time 
to $\CSP(\struct{D})$. 
\end{itemize} 
\end{theorem}
Since there are finite-domain bounded alternation  
QCSPs  complete for the $\Sigma^{\PO}_{2n+1}$- and $\Pi^{\PO}_{2n}$-levels of the polynomial hierarchy, Theorem~\ref{thm:qcsp_to_csp} has the following consequence.
\begin{theorem} \label{thm:odd_part}
    For every $n \geq 1$, there are 
    $\Sigma_{2n+1}^\PO$- and $\Pi_{2n}^\PO$-complete $\omega$-categorical CSPs. 
\end{theorem}

We cannot use Theorem~\ref{thm:qcsp_to_csp} to obtain complete problems for the
$\Sigma^\PO_{2n}$- or $\Pi^\PO_{2n+1}$-levels  of the \textrm{PH}
because finite-domain  $(\exists\forall )^n$QCSPs cannot ever 
be $\Sigma^\PO_{2n}$- or $\Pi^\PO_{2n+1}$-complete. 
We then prove the following analogue of Theorem~\ref{thm:odd_part} for the
remaining levels.
\begin{theorem} \label{thm:even_part}
      For every $n \geq 1$, there are $\Sigma_{2n}^\PO$-  and $\Pi_{2n+1}^\PO$-complete
      $\omega$-categorical CSPs. 
\end{theorem} 
The existence of $\Pi_n^{\PO}$-complete $\omega$-categorical CSPs
was originally proved in~\cite{bodirsky2008non}. However, our methodology
is robust enough to naturally yield $\Pi_n^{\PO}$-complete $\omega$-categorical CSPs while
constructing $\Sigma_n^{\PO}$-complete
$\omega$-categorical CSPs.  
These results, and the fact that there are
\textrm{NP}-complete finite-domain CSPs (e.g., 3-\textsc{Sat}), together prove our main theorem.
 \begin{theorem}
    \label{thm:completeness}
For every level $\Sigma_{n}^\PO$ or $\Pi_{n}^\PO$ $(n\geq 1)$ in the polynomial hierarchy, there exists an $\omega$-categorical structure $\struct{D}$ such that $\CSP(\struct{D})$ is complete for that level. 
 \end{theorem}

\section{Preliminaries}\label{sect:prelims}  

We use the bar notation $\bar .$ for tuples, and denote the set $\{1,\dots,n\}$ by $[n]$. 
We extend the membership relation on sets to tuples by ignoring the ordering on the entries. For example, we might write $x\in \bar{x}$ to denote that $x$ is an entry in $\bar{x}$. 
 
\subsection{Structures.} 
A (\emph{relational}) \emph{signature} $\tau$ is a set of \emph{relation symbols}, where each $R\in\tau$ is associated with a natural number called \emph{arity}.
A (\emph{relational}) \emph{$\tau$-structure} $\struct{A}$ consists of a set $A$ (the \emph{domain}) together with the relations $R^{\struct{A}}\subseteq A^{k}$ for each $R\in \tau$ with arity $k$.
An \emph{expansion} of $\struct{A}$ is a $\sigma$-structure $\struct{A}^+$ with $A=A^+$ such that $ \tau\subseteq \sigma$ and $R^{\struct{A}^+}=R^{\struct{A}}$ for each relation symbol $R\in \tau$. Conversely, we then call $\struct{A}$ a \emph{reduct} of $\struct{A}^+$.
We denote by $(\struct{A},R)$ the expansion of a structure $\struct{A}$ by a relation $R$ over its domain. 
The \emph{substructure} of a $\tau$-structure $\struct{B}$ on a subset $A\subseteq B$ is the $\tau$-structure $\struct{A}$ with domain $A$ and relations $R^{\struct{A}}=R^{\struct{B}}\cap A^k$ for every $R\in \tau$ of arity $k$. 
Conversely, we call $\struct{B}$ a \emph{superstructure} of $\struct{A}$.
%
%
The \emph{union} of two $\tau$-structures $\struct{A}$ and $\struct{B}$ is the $\tau$-structure $\struct{A}\cup \struct{B}$ with domain $A\cup B$ and relations $R^{\struct{A}\cup \struct{B}}\coloneqq R^{\struct{A}}\cup R^{\struct{B}}$ for every $R\in \tau$.
If $A\cap B = \emptyset$, we call $\struct{A}\cup \struct{B}$ a \emph{disjoint union} and write $\struct{A}+\struct{B}$. 
%
%
%
A \emph{homomorphism} $h\colon \struct{A} \rightarrow \struct{B}$ for $\tau$-structures $\struct{A},\struct{B}$ is a mapping $h\colon  A\rightarrow B$ that \emph{preserves} each relation of $\tau$, i.e., whenever $ \bar{t} \in R^{\struct{A}}$ for some  relation symbol $R\in \tau$, then $h(\bar{t})$ (computed componentwise)  is an element of $R^{\struct{B}}$.
We write $\struct{A} \rightarrow \struct{B}$ if $\struct{A}$ maps homomorphically into $\struct{B}$.  
A class $\mathcal{K}$ of finite structures is \emph{preserved} or \emph{closed} under:
\begin{itemize}
    \item \emph{homomorphisms} if, for every $\struct{A} \in \mathcal{K}$, every finite structure $\struct{B}$ with a homomorphism from $\struct{A}$ is also contained in $\mathcal{K}$;
    \item \emph{inverse homomorphisms} if, for every $\struct{B} \in \mathcal{K}$, every finite structure $\struct{A}$ with a homomorphism to $\struct{B}$ is also contained in $\mathcal{K}$;
    \item \emph{disjoint unions} if, for all $\struct{A}, \struct{B}\in \mathcal{K}$, the union of two isomorphic copies of $\struct{A}$ and $\struct{B}$ with disjoint domains is also contained in $\mathcal{K}$.
\end{itemize}

\subsection{Logic.} 
We assume that the reader is familiar with classical \emph{First-Order} logic FO.   
A first-order $\tau$-formula $\phi$ is \emph{primitive positive} (pp) if it is of the form $\exists x_{1},\dots,x_{m}  (\phi_{1}\wedge \dots \wedge \phi_{n})$, where each $\phi_{i}$ is \emph{atomic}, i.e., of the form $\bot$, $(x_{j}=x_{k})$, or $R(x_{i_{1}},\dots,x_{i_{\ell}})$ for some $R\in \tau$. 
\emph{Quantified primitive positive} (qpp) formulas generalize pp-formulas by allowing both existential \emph{and} universal quantification. 
We use the stylized letter $\mathrm{Q}$ for the quantifiers $\exists$ and $\forall$; given $\mathrm{Q}\in \{\exists,\forall\}$, we denote the complementary quantifier by $\overline{\mathrm{Q}}$.

\emph{Second-Order logic} SO is the extension of FO which additionally allows existential and universal quantification over relations; that is, if $R$ is a relation symbol and $\phi$ is a SO $\tau\cup \{R\}$-formula, then $\exists R \ldotp \phi$ and $\forall R \ldotp \phi$ are SO $\tau$-formulas. 
\emph{Monadic Second-Order logic} MSO is the restriction of SO to unary SO variables.
In $\exists R \ldotp \phi$ and $\forall R \ldotp \phi$, we then refer to $R$ as a \emph{SO variable}. 
The \emph{model-checking problem} for a fixed $\tau$-formula $\phi$ is the computational problem of deciding whether $\phi$ is satisfiable in a given finite $\tau$-structure $\struct{A}$.

We say that a formula $\phi$ is \emph{$k$-ary} if it has $k$ free FO variables; we use the notation $\phi(\bar{x})$ to indicate that the free FO variables of $\phi$ are among $\bar{x}$.
This does not mean that the truth value of $\phi$ depends on each entry in $\bar{x}$. 

A \emph{sentence} is a formula without free FO variables; we reserve uppercase letters for sentences and lowercase for formulas.
We use $\fm(\Phi)$ to denote the class of all finite models of a sentence $\Phi$; when we say that $\Phi$ is closed under (inverse) homomorphisms or disjoint unions, we mean that $\fm(\Phi)$ has the respective property.

We say that a formula $\phi$ is \emph{normalized} if it is in prenex normal form.
By this we mean that $\phi$ starts with an SO quantifier-prefix, followed by an FO quantifier-prefix, and ends with a quantifier-free formula.
We say that $\phi$ is \emph{CNF-} or \emph{DNF-normalized} if it is normalized and its quantifier-free part is in \emph{Conjunctive Normal Form} (CNF) or \emph{Disjunctive Normal Form} (DNF), respectively.
With a simple inductive argument over the definition of SO, one can verify that every SO formula is logically equivalent to a CNF- and a DNF-normalized one.

We say that a class $\mathcal{L}$ of SO formulas is \emph{closed under internal conjunctions and disjunctions} if the following holds:
if $\phi_1$ and $\phi_2$ are $\tau\cup \{S_1,\dots, S_n\}$-formulas and $\mathrm{Q}_1,\dots,\mathrm{Q}_n$ quantifiers such that the $\tau$-formulas $\mathrm{Q}_1 S_1 \cdots \mathrm{Q}_n S_n \ldotp \phi_1$ and $\mathrm{Q}_1 S_1 \cdots \mathrm{Q}_n S_n \ldotp \phi_2$ are contained in $\mathcal{L}$, then
$\mathrm{Q}_1 S_1 \cdots \mathrm{Q}_n S_n\big( \phi_1 \wedge \phi_2\big)$ and $\mathrm{Q}_1 S_1 \cdots \mathrm{Q}_n S_n\big( \phi_1 \vee \phi_2\big)$ are contained in $\mathcal{L}$ as well. 
For a CNF- or DNF-normalized SO formula $\phi$, we say that $\phi$ is \emph{positive} (\emph{negative}) if atomic $\tau$-formulas of the form $R(\bar{x})$ for $R\in \tau$ may only appear positively (negatively) in clauses of $\phi$, i.e., every negative (positive) atom contains a SO variable or the default equality predicate.  
We \emph{naturally} extend this definition to normalized SO formulas $\phi$ by instead requiring the condition to be satisfied by some formula $\phi'$ obtained from $\phi$ by converting the quantifier-free part of $\phi$ into CNF using the De Morgan's laws and the distributive law of propositional logic.
Note that we could have equivalently used DNF instead of CNF.
Finally, we define the sets of \emph{positive} (\emph{negative}) \emph{SO formulas} as the closure of positive (negative) normalized SO formulas under internal conjunctions and disjunctions. 
Alternatively, we could define positivity and negativity in terms of $\tau$-atoms being under the scope of evenly and oddly many negations, respectively. 
However, the current definition is more suitable for our applications.
\begin{example} \label{ex:acyclic}
The following MSO $\{E\}$-sentence $\Phi$, testing whether a binary relation $E$ contains a  directed cycle, is positive:
\begin{align*}
     \exists S \, \forall C \big(\exists x\ldotp E(x,x)\big) \vee \Big(\big(\exists x,y \ldotp S(x) \wedge S(y) \wedge \neg (x = y) \big) \\ {} \wedge \Big(    \big(\exists a,b   \ldotp    S(a) \wedge S(b) \wedge C(a)   \wedge E(a,b)  \wedge \neg C(b) \big)  \\ {}  \vee\big(
     \forall x,y\ldotp  \neg S(x) \vee \neg S(y) \vee   \neg C(x)  \vee C(y)
    \big)\Big) \Big).
\end{align*} 
%
%
Negating $\Phi$ and applying De Morgan's laws instead yields a negative sentence.
\end{example}  
 
It is useful to take the following game-theoretic perspective at SO.
The evaluation of a SO $\tau$-formula $\phi$ on a structure $\struct{A}$ corresponds to a game between two players: an \emph{existential player} (EP) and a \emph{universal player} (UP) who assign values to the existentially and universally quantified variables, respectively.
To every moment of the game we associate a partial function $\ass{\cdot}$   describing values assigned to the variables of $\phi$ by either of the players:
\begin{itemize}
    \item if $x$ is a FO variable, then $\ass{x} $ is either undefined or an element of $A$; 
    \item if $S$ is a SO variable of arity $k$, then $\ass{S} $ is either undefined or a subset of $A^k$. 
\end{itemize}  
The goal of the EP is to satisfy the quantifier-free part of $\phi$ evaluated at $\bar a$; the goal of the UP is the opposite.
Clearly, $\struct{A}\models \phi(\bar{a})$ if  the EP has a winning strategy in this game, i.e., can respond to all moves of the UP while keeping the quantifier-free part of $\phi$ evaluated at $\bar a$ satisfied. 
Otherwise, $\struct{A} \centernot{\models} \phi(\bar{a})$ and the UP has a winning strategy, i.e., can violate the quantifier-free part of $\phi$ evaluated at $\bar a$ regardless of the moves of the EP.

\subsection{Constraint satisfaction problems.}
The \emph{(Quantified) Constraint Satisfaction Problem} for a structure $\struct{B}$, denoted by $\mathrm{(Q)}\CSP(\struct{B})$, is the computational problem of deciding whether a given 
(q)pp $\tau$-sentence holds in $\struct{B}$. 
By \emph{constraints}, we refer to the conjuncts in the quantifier-free part of a given instance of $\mathrm{(Q)}\CSP(\struct{B})$.
In the CSP framework, we also consider the homomorphism perspective, where an instance is viewed as a finite $\tau$-structure.
Then $\CSP(\struct{B})$ is (the membership problem for) the class of all finite structures $\struct{A}$ which homomorphically map into $\struct{B}$.
The two definitions are equivalent~\cite{10.1145/800105.803397}. 
%
%
\emph{Bounded alternation} QCSPs are QCSPs restricted to a fixed quantifier prefix $\mathrm{Q}_1\dots \mathrm{Q}_n$.
In particular, CSPs are bounded alternation QCSPs restricted to the quantifier prefix $\exists$.
Finite-domain \emph{bounded alternation} QCSPs have been investigated in~\cite{chen2009existentially}.

We usually think of an instance of a QCSP as a game between two
players: an \emph{existential player} (EP) and a \emph{universal player} (UP) who assign values to the existentially and universally quantified variables, respectively.
To every moment of the game we associate a partial function $\ass{\cdot}$ from the variables into the domain of the parametrizing structure $\struct{B}$ describing values assigned to the variables by either of the players.
The instance is true if and only if the EP has a winning strategy in this game, i.e., can respond to all moves of the UP while keeping all constraints satisfied. Otherwise, the instance is false and the UP has a winning strategy, i.e., can violate a constraint regardless of the moves of the EP. 

\subsection{3-CNF}
The \emph{satisfiability problem} for 3-CNF asks whether a given propositional formula in conjunctive normal form with at most three variables per clause has a satisfying assignment. 
It can naturally be represented as a finite-domain CSP by interpreting each clause as a ternary relational constraint over $\{0,1\}$, e.g., $(\neg x_i \vee x_j \vee x_k)$ as $R_{100}(x_i,x_j,x_k)$, where $R_{100}\coloneqq \{0,1\}^3 \setminus \{(1,0,0)\}.$
Allowing existential \emph{and} universal quantification over the input variables yields the computational decision problem \textsc{Quantified 3-CNF}.

\section{Preservation lemma}  \label{section:syntactic}
 
We begin this section by giving a syntactic characterization of the fragment of SO closed under injective homomorphisms (Lemma~\ref{thm:exist_guarded_positive}).  
We then dualize the said syntactic characterization to obtain a one that characterizes the fragment of SO closed under inverse injective
homomorphisms (Lemma~\ref{thm:exist_guarded_negative}). 
In Section~\ref{subsect:hammer}, we convert the dualized syntactic characterization into a tool for producing SO sentences closed under inverse homomorphisms and disjoint unions. 
This tool will be used in Sections~\ref{sect:odd_part} and~\ref{sect:even_part} to prove our main results, Theorems~\ref{thm:odd_part} and~\ref{thm:even_part}.

\subsection{Injective homomorphisms}

In Lemma~\ref{thm:exist_guarded_positive}, we give a syntactic condition for SO sentences which implies closure under injective homomorphisms.
Bodirsky, Feller, Kn\"{a}uer, and Rudolph~\cite[Theorem~59]{bodirsky2021logics} recently identified a syntactic transformation $\Phi \mapsto (\Phi^{\mathrm{shom}})^{\mathrm{sup}}$ for SO sentences with the property that $\fm(\Phi)=\fm((\Phi^{\mathrm{shom}})^{\mathrm{sup}})$ if and only if $\Phi$ is preserved under homomorphisms. 
We can show that, for every CNF-normalized SO sentence $\Phi$, the sentence  $(\Phi^{\mathrm{shom}})^{\mathrm{sup}}$ satisfies our syntactic condition for preservation under injective homomorphisms, and hence our condition captures the fragment of SO closed under injective homomorphisms up to logical equivalence.  
 \begin{definition} \label{def:exist_guard}
     A CNF-normalized SO sentence $\Phi$ is \emph{$\exists$-guarded} if, for every clause $\phi$ in $\Phi$ and every universally quantified FO variable $x$ in $\phi$, there is a disjunct $\neg S(\bar{x})$ in $\phi$ for some existentially quantified SO variable $S$  such that $x \in \bar x$. 
     We naturally extend this definition to normalized SO sentences.
 The set of \emph{$\exists$-guarded SO sentences} is  the closure of $\exists$-guarded normalized SO sentences under internal conjunctions and disjunctions.  
 \end{definition}

It is easy to verify that the sentence in Example~\ref{ex:acyclic} is $\exists$-guarded; we only need to break it down to its CNF-normalized components and apply Definition~\ref{def:exist_guard} inductively.

Intuitively, the condition of $\exists$-guarding allows the EP to restrict the game of  
evaluating a SO sentence in a structure $\struct A$ to any particular substructure $\struct B$; all the EP has to do is to evaluate 
every existentially quantified SO variable within $B$. 
This reasoning can be used to show that every $\exists$-guarded SO sentence is preserved under superstructures,
because if  $\struct B$ satisfies the SO sentence, then the EP can choose to restrict the game on $\struct B$.
Corollary~55 in~\cite{bodirsky2021logics} provides a transformation $\Phi \mapsto \Phi^{\mathrm{sup}}$ such that $\fm(\Phi)=\fm(\Phi^{\mathrm{sup}})$ if and only if $\Phi$ is preserved under superstructures.
The following lemma extends~\cite[Corollary~55]{bodirsky2021logics} to a syntactic characterization of the SO-fragment closed under superstructures.
 
\begin{restatable}{lemma}{exguarded}   \label{lem:ex_guarded}
    For a SO sentence $\Phi$, the following are equivalent:
 \begin{enumerate}
     \item \label{item:ex1} $\Phi$ is preserved under superstructures.
     \item \label{item:ex2} $\Phi$ is logically equivalent to an $\exists$-guarded SO sentence.
 \end{enumerate} 
 
\end{restatable}   
\begin{proof}   We only prove the statement for normalized SO sentences; it will become clear that the argument extends to the full inductive definition. Let $\tau$ be the signature of $\Phi$.  

``\eqref{item:ex2}$\Rightarrow$\eqref{item:ex1}'' 
We may assume that $\Phi$ itself is CNF-normalized; otherwise we replace its quantifier-free part with the CNF version of it witnessing $\exists$-guarding. 
%
Suppose that $\struct{A} \models \Phi$ and $\struct{B}$ is a superstructure of
$\struct{A}$.
We claim that the EP has the following winning strategy on
$\struct{B}$ based on the winning strategy on $\struct{A}$.
For a SO variable $S$:
    \begin{itemize}
        \item if $S$ is existentially quantified and the EP chooses
        $X\subseteq A^k$ for $\ass{S}$ on $\struct{A}$, then the EP chooses
        $X\subseteq A^k\subseteq B^k$ for $\ass{S}$ on $\struct{B}$;  
        \item if $S$ is universally quantified and the UP reacts by
        choosing $Y\subseteq B^k$ for $\ass{S}$ on $\struct{B}$, then we consider
        the situation where the UP chooses $ Y\cap A^k \subseteq A^k$
        for $\ass{S}$ on $\struct{A}$; 
    \end{itemize}
For FO variables, the EP follows the winning strategy on $\struct{A}$ and interprets any responses by the UP with values outside of $A$ as arbitrary responses inside $A$.

Suppose that, when playing on $\struct{B}$, the UP chooses $\ass{x}\in B\setminus A$. 
Since $\Phi$ is $\exists$-guarded, every clause $\phi$ contains a disjunct
$\neg S(\bar{x})$ in $\phi$ for some existentially quantified SO variable $S$.
By the strategy of the EP, we have that $\ass{S} \subseteq A^k$ for some $k$.
Hence, this choice of the UP trivially satisfies $\phi$. 
Alternatively, the strategy of the UP on $\struct{B}$ simulates a strategy
on $\struct{A}$. Hence, the whole game on $\struct B$ simulates a game 
on $\struct A$, and by the choice of strategy for the EP, we conclude
that the EP wins.

``\eqref{item:ex1}$\Rightarrow$\eqref{item:ex2}''  
This direction follows from~\cite[Corollary~55]{bodirsky2021logics} using $\Phi \mapsto \Phi^{\mathrm{sup}}$.
We give a full proof for the convenience of the reader.
We may assume that $\Phi$ is CNF-normalized; otherwise we replace it with any logically equivalent CNF-normalized sentence.
Denote by $\Phi^{\mathrm{sup}}$ the SO $\tau$-sentence obtained from $\Phi$ by adding:
\begin{itemize}
    \item a fresh SO variable $S$ existentially quantified at the beginning of the SO quantifier prefix,
    \item a fresh FO variable $z$ existentially quantified at the beginning of the FO quantifier prefix, together with the clause $S(z)$,
    \item for every existential FO variable $y$, a clause $\psi$ that contains a disjunct $\lnot S(x)$ for every
    universal FO variable $x$, and a disjunct $S(y)$,
    \item for every  universal FO  variable $x$, the disjunct $\neg S(x)$ to every
    clause in the quantifier-free part of $\Phi$.
\end{itemize}
Clearly, $\Phi^{\mathrm{sup}}$ is an $\exists$-guarded sentence. 
%
Moreover, there exists a non-empty substructure $\struct{B}$ of $\struct{A}$ such that $\struct{B} \models \Phi$ if and only if $\struct{A} \models \Phi^{\mathrm{sup}}$. 
Since $\Phi$ is preserved under superstructures, we have that $\Phi$ and $\Phi^{\mathrm{sup}}$ are logically equivalent.
\end{proof}

Next, Corollary~58 in~\cite{bodirsky2021logics} provides a transformation $\Phi \mapsto \Phi^{\mathrm{shom}}$ such that $\fm(\Phi)=\fm(\Phi^{\mathrm{shom}})$ if and only if $\Phi$ is preserved under surjective homomorphisms.
The following lemma provides a syntactic characterisation of the SO-fragment closed under bijective homomorphisms.
It is compatible with $\Phi \mapsto \Phi^{\mathrm{shom}}$ in the sense that $\Phi^{\mathrm{shom}}$ is positive whenever $\Phi$ is a normalized SO sentence, but it does not extend~\cite[Corollary~58]{bodirsky2021logics} in its full generality.
Recall the definition of positivity for SO formulas from Section~\ref{sect:prelims}. 
 
\begin{restatable}{lemma}{pos}   \label{lem:pos}
 For a SO sentence $\Phi$, the following are equivalent:
 \begin{enumerate}
     \item \label{item:pos1} $\Phi$ is preserved under bijective homomorphisms.
     \item \label{item:pos2} $\Phi$ is logically equivalent to a positive SO sentence.
 \end{enumerate} 
 
\end{restatable}     
\begin{proof}  
We only prove the statement for normalized SO sentences; it will become clear that the argument extends to the full inductive definition. 
 Let $\tau$ be the signature of $\Phi$.

``\eqref{item:pos2}$\Rightarrow$\eqref{item:pos1}'' 
We may assume that $\Phi$ itself is CNF-normalized; otherwise we replace its quantifier-free part with the CNF version of it witnessing positivity. 
Let $\struct{A}$ and $\struct{B}$ be two finite $\tau$-structures over the signature $\tau$ of $\Phi$ such that $\struct{A} \models \Phi$ and there exists a bijective homomorphism $h\colon \struct{A} \rightarrow \struct{B}$. Without loss of generality we assume that $A = B$, and that $R^\struct A \subseteq R^\struct B$
for all $R\in \tau$.
    We claim that the EP has the following winning strategy on $\struct{B}$ based on the winning strategy on $\struct{A}$:
    For a SO variable $S$:
    \begin{itemize}
        \item if $S$ is existentially quantified and the EP chooses $X\subseteq A^k$ for $\ass{S}$ on $\struct{A}$, then the EP chooses $X\subseteq B^k = A^k$
        for $\ass{S}$ on $\struct{B}$;  
        \item if $S$ is universally quantified and the UP reacts by choosing $Y\subseteq B^k$ for $\ass{S}$ on $\struct{B}$, then we consider the situation
        where the UP chooses $ Y \subseteq A^k = B^k$ for $\ass{S}$ on $\struct{A}$.
    \end{itemize} 
    The FO variables are treated analogously.
    
    Denote by $\smash{\struct{A}^+}$ and $\struct{B}^{\smash{+}}$ the expansions of $\struct{A}$ and $\struct{B}$ resulting from the above evaluations of the SO variables in $\Phi$.  
    Let $\big(\phi_1\vee \cdots \vee \phi_{k}\big)$ be a clause in the quantifier-free part of $\Phi$.
    Since the strategy of the EP on $\struct{A}$ is a winning strategy, there exists $i\in [k]$ such that $\smash{\struct{A}^+} \models \phi_i(\ass{\bar{x}})$. 
    For every SO variable $S$, we have that $S^{\struct{A}^+} = S^{\struct{B}^+}$.  %
    Hence, if $\phi_i(\bar{x})$ contains an SO variable $S$, i.e., $\phi_i(\bar{x}) \in\{S(\bar x), \neg S(\bar{x})\}$, then $\struct{B}\models \phi_i(h(\ass{\bar{x}}))$ is immediate. 
    Also, if $\phi_i$ is of the form $(x_1=x_2)$ or $\neg (x_1=x_2)$, then $\struct{B}\models \phi_i(h(\ass{\bar{x}}))$ is immediate because $h$ is bijective.
    Otherwise, since $\Phi$ is positive, $\phi_i(\bar{x})$ must be of the form $R(\bar{x})$ for $R\in \tau $.  
    Then $\struct{B}\models \phi_i(h(\ass{\bar{x}}))$ follows
    because $R^\struct{A}\subseteq R^\struct{B}$.  

``\eqref{item:pos1}$\Rightarrow$\eqref{item:pos2}''
This direction follows from~\cite[Corollary~58]{bodirsky2021logics} using $\Phi \mapsto \Phi^{\mathrm{shom}}$.
We give a full proof for the convenience of the reader.
We may assume that $\Phi$ is CNF-normalized; otherwise we replace it with any logically equivalent CNF-normalized sentence.
For every $R\in \tau $, we introduce a SO variable $R_p$ of the same arity as $R$.
Now consider the positive SO sentence $\Phi^{\mathrm{shom}}$ defined as 
 \[
 \exists_{R\in \tau } R_p \Big( \Phi'_p \wedge \bigwedge\nolimits_{R\in \tau } \forall \bar{x} \big(\neg R_p(\bar{x}) \vee R(\bar{x})   \big) \Big),
\]
where $\Phi'_p$ is obtained from $\Phi$ by replacing every atom of the form $R(\bar{x})$ for $R\in \tau$ with $R_p(\bar{x})$. 
Clearly, $\Phi^{\mathrm{shom}}$ is positive; we show that $\Phi$ and $\Phi^{\mathrm{shom}}$ are logically equivalent.
Let $\struct{A}$ be a finite $\tau$-structure. If $\struct{A}\models \Phi$, then clearly $\struct{A} \models \Phi_{p}$; the winning strategy 
is to first choose $R^{\struct{A}}$ for every $R_p$ ($R\in \tau$)  and mimic the winning strategy witnessing that $\struct A\models \Phi$.

Now suppose that $\struct{A}\models \Phi^{\mathrm{shom}}$.
For every $R\in \tau$,  let $R_p^\struct{A}$
be the evaluation of the quantified predicate $R_p$ in a winning strategy for the EP.
Then clearly,  the structure  $\struct{A}_p \coloneqq (A, (R_p^\struct{A})_{R\in\tau})$ satisfies
$\Phi$. Since $\Phi^{\mathrm{shom}}$ entails  $\forall \bar{x} \big(\neg R_p(\bar{x}) \vee R(\bar{x})   \big)$ for every $R\in \tau$,
the identity map is a bijective homomorphism from $\struct{A}_p$ to $\struct{A}$. Since $\Phi$ is preserved by bijective homomorphisms,
we have that $\struct{A} \models \Phi$. 
\end{proof}

 The transformations $\Phi \mapsto \Phi^{\mathrm{sup}}$, $\Phi \mapsto \Phi^{\mathrm{shom}}$ from~\cite{bodirsky2021logics} used to prove ``(1)$\Rightarrow$(2)'' in Lemmas~\ref{lem:ex_guarded},~\ref{lem:pos} can be combined, yielding an $\exists$-guarded positive sentence. This observation together with the
 implication ``(2)$\Rightarrow$(1)'' from these lemmas, yield the following result.

\begin{restatable}{lemma}{existguardedpositive}  \label{thm:exist_guarded_positive}
     For a SO sentence $\Phi$, the following are equivalent:
 \begin{enumerate}
     \item \label{item:pos1ex} $\Phi$ is preserved under injective homomorphisms.
     \item \label{item:pos2ex} $\Phi$ is logically equivalent to an $\exists$-guarded positive SO sentence.
     \item \label{item:pos3ex} $\Phi$ is logically equivalent to an $\exists$-guarded SO sentence and a positive SO sentence.
 \end{enumerate}
\end{restatable} 
\begin{proof}
``\eqref{item:pos2ex}$\Rightarrow$\eqref{item:pos3ex}'' This direction is trivial.

``\eqref{item:pos3ex}$\Rightarrow$\eqref{item:pos1ex}'' This direction is by a direct application of Lemmas~\ref{lem:ex_guarded} and~\ref{lem:pos}.

``\eqref{item:pos1ex}$\Rightarrow$\eqref{item:pos2ex}''  We may assume that $\Phi$ is CNF-normalized; otherwise we replace it with any logically equivalent CNF-normalized sentence.
Let  $\Phi \mapsto \Phi^{\mathrm{sup}}$ be the transformation from the proof
of ``\eqref{item:ex1}$\Rightarrow$\eqref{item:ex2}'' in Lemma~\ref{lem:ex_guarded}, and $\Phi \mapsto \Phi^{\mathrm{shom}}$
the transformation from the proof of ``\eqref{item:pos1}$\Rightarrow$\eqref{item:pos2}'' in Lemma~\ref{lem:pos}. 
It is straightforward to observe that $(\Phi^{\mathrm{shom}})^{\mathrm{sup}}$ is $\exists$-guarded
positive SO sentence. And since $\Phi$ is preserved under inverse injective homomorphisms, we conclude
with similar arguments as in the proofs of ``(1)$\Rightarrow$(2)'' in Lemmas~\ref{lem:ex_guarded}
and~\ref{lem:pos} that $\Phi$ is logically equivalent to $(\Phi^{\mathrm{shom}})^{\mathrm{sup}}$. 
 \end{proof}

\subsection{Inverse injective homomorphisms} 
 \label{section:dual}
Note that $\mathcal{K}$ is closed under inverse injective homomorphisms if and only if the \emph{complement} of $\mathcal{K}$, i.e., the class of all finite $\tau$-structures not in $\mathcal{K}$, is closed under injective homomorphisms.
Hence, by Lemma~\ref{thm:exist_guarded_positive}, $ \Phi$ is logically equivalent to an $\exists$-guarded positive SO sentence if and only if  $\neg  \Phi$ is closed under inverse injective homomorphisms.
The following fact can be  easily  shown using the De Morgan's laws and the distributive law of propositional logic.
\begin{lemma}   \label{lem:dualizing_positivity}
    For every SO sentence $\Phi \coloneqq \mathrm{Q}_1 S_1\cdots \mathrm{Q}_m S_m \ldotp \Phi_{\ast}$, the following are equivalent:
 \begin{enumerate}  
     \item    There is a positive SO sentence $\mathrm{Q}_1 S_1\cdots \mathrm{Q}_m S_m \ldotp \Psi_{\exists}$ logically equivalent to $\Phi$.
    \item   There is a negative SO sentence $\overline{\mathrm{Q}}_1 S_1\cdots \overline{\mathrm{Q}}_m S_m \ldotp \Psi_{\forall}$ logically equivalent to $\neg \Phi$. 
 \end{enumerate}  
\end{lemma}

We can therefore say that negativity constitutes the \emph{dual} notion to positivity. 
We shall now dualize the notion of $\exists$-guarding.
 \begin{definition}\label{def:forallguarding} 
     A DNF-normalized SO sentence $\Phi$ is \emph{$\forall$-restricted} if, for every clause $\phi$ in $\Phi$ and every existentially quantified FO variable $x$ in $\phi$, there is a conjunct $S(\bar{x})$ in $\phi$ for some universally quantified SO variable $S$ such that $x \in \bar x$. 
     We naturally extend this definition to normalized SO sentences.
 The set of \emph{$\forall$-restricted SO sentences} is the closure of $\forall$-restricted normalized SO sentences under internal conjunctions and disjunctions.  
 \end{definition}

%
Analogously to $\exists$-guarding, the condition of $\forall$-restriction allows the UP to restrict 
the game of evaluating a SO sentence in a structure to any particular substructure 
by  evaluating every universally quantified SO variable within this substructure.
Hence, the UP has the power to verify that the $\forall$-restricted SO sentence is satisfied on all
substructures, and this reasoning can be used to show that every $\forall$-restricted SO sentence is preserved under substructures.
Instead of giving the obvious counterpart to Lemma~\ref{lem:pos}, we give a lemma that
captures the  duality between $\exists$-guarding and $\forall$-restriction.

\begin{restatable}{lemma}{duality}\label{lem:dualizing_exists_guarding}
  For every SO sentence $\Phi \coloneqq \mathrm{Q}_1 S_1\cdots \mathrm{Q}_m S_m \ldotp \Phi_{\ast}$, the following are equivalent:
 \begin{enumerate}  
     \item \label{item:dual1} There is an $\exists$-guarded SO sentence $\mathrm{Q}_1 S_1\cdots \mathrm{Q}_m S_m \ldotp \Psi_{\exists}$ logically equivalent to $\Phi$.
    \item \label{item:dual2}  There is an $\forall$-restricted SO sentence $\overline{\mathrm{Q}}_1 S_1\cdots \overline{\mathrm{Q}}_m S_m \ldotp \Psi_{\forall}$ logically equivalent to $\neg \Phi$. 
 \end{enumerate}  
 
\end{restatable}   
\begin{proof}  
We only prove the statement for normalized SO sentences. 

    ``\eqref{item:dual1}$\Rightarrow$\eqref{item:dual2}'' 
    Without loss of generality, $\mathsf{Q}_1 S_1\cdots \mathsf{Q}_m S_m \ldotp \Phi_{\exists}$ is CNF-normalized; otherwise we replace its quantifier-free part with the CNF version of it witnessing $\exists$-guarding.
By applying De Morgan's laws, we get that $\neg(\mathsf{Q}_1 S_1\cdots \mathsf{Q}_m S_m \ldotp \Phi_{\exists})$ is logically equivalent to a SO sentence $\overline{\mathsf{Q}}_1 S_1\cdots \overline{\mathsf{Q}}_m S_m \ldotp \Phi_{\forall}$ in prenex normal form and whose quantifier-free part is in disjunctive normal form.
Since $\mathsf{Q}_1 S_1\cdots \mathsf{Q}_m S_m \ldotp \Phi_{\exists}$ is $\exists$-guarded, for every clause $\phi$ in $\Phi_{\forall}$ and every existentially quantified FO variable
$x$ which appears in $\phi$, there is a conjunct $S(\bar{x})$
in $\phi$ for some universally quantified SO variable
$S$, and $x$ is an argument of $\bar x$.  

``\eqref{item:dual2}$\Rightarrow$\eqref{item:dual1}'' This direction is analogous. 
\end{proof} 
 
Note that Lemmas~\ref{lem:dualizing_positivity} and~\ref{lem:dualizing_exists_guarding} are formulated in such a way that allows us to exploit the inductive nature of the involved definitions.
For example, to check whether a given sentence is $\exists$-guarded, we can first decompose it into its normalized components
with the same SO quantifier prefix, and then, for each component we check $\exists$-guarding, or verify $\forall$-restriction for its negation.
The important part is that, after the application of Lemma~\ref{lem:dualizing_exists_guarding}, we keep the original SO quantifier prefix.
Finally, combining Lemma~\ref{lem:dualizing_positivity} and Lemma~\ref{lem:dualizing_exists_guarding},
we obtain the following dual statement as a corollary to Lemma~\ref{thm:exist_guarded_positive}. 

\begin{lemma}
    \label{thm:exist_guarded_negative}
     For a SO sentence $\Phi$, the following are equivalent:
 \begin{enumerate}
     \item   $\Phi$ is preserved under inverse injective homomorphisms.
     \item   $\Phi$ is logically equivalent to a $\forall$-restricted negative SO sentence.
     \item   $\Phi$ is logically equivalent to a $\forall$-restricted SO sentence and a negative SO sentence.  
 \end{enumerate}
\end{lemma}

\section{The CSP hammer}
\label{subsect:hammer}

Recall from the paragraph above Lemma~\ref{lem:pos} that the notion of positivity characterizing the preservation by bijective homomorphisms is compatible with the transformation $\Phi \mapsto \Phi^{\mathrm{shom}}$ from~\cite[Corollary~58]{bodirsky2021logics} characterizing preservation by surjective homomorphisms.
More specifically, $\Phi^{\mathrm{shom}}$ is positive whenever $\Phi$ is a CNF-normalized SO sentence.
However, not every positive SO sentence is preserved by surjective homomorphisms; dually, not every
negative SO sentence is preserved under inverse surjective homomorphisms. 
%
%
\begin{example} \label{ex:positivity_sucks}
For a binary symbol $N$, consider the following MSO $\{N\}$-sentence:
\begin{align*}
  \Phi \coloneqq  \forall A \, \exists E\, \forall x,y \Big(\big( \neg N(x,y) \vee \neg A(x) \vee E(y) \big) \wedge \big( \neg N(x,y) \vee \neg E(x) \vee A(y) \big)  \Big).   
\end{align*} 
This sentence simply says that the EP has to copy the $A$-colouring of vertices by the UP with an $E$-colouring across $N$-edges.
Clearly, $\Phi$ is CNF-normalized, negative because $\{N\}$-atoms only appear negatively, and $\forall$-restricted because it does not contain any existentially quantified FO variables.
Hence, it is preserved under inverse injective homomorphisms (Lemma~\ref{thm:exist_guarded_negative}).
However, it is not preserved under inverse surjective homomorphisms.
Indeed, the graph $\struct{G}\coloneqq ([3];\{(1,2),(2,1),(3,2),(2,3)\})$ has a surjective homomorphism to the 2-clique $\struct{K}_2=([2];\{(1,2),(2,1)\})$.
On $\struct{K}_2$, the EP wins the evualuation game for $\Phi$ because every vertex only has a single neighbour.
On $\struct{G}$, however, the UP has a winning strategy by choosing to $A$-colour $1$ while not $A$-colouring $3$. 
\end{example}

It is not clear for us how to generalize the notion of positivity (negativity) 
so that it would capture the full fragment of SO closed under (inverse) surjective homomorphisms.
%
By extension, we face the same obstacle to identify a fragment of SO closed
that characterizes preservation under inverse homomorphisms; it is also worth noting that testing whether
a given SO sentence is preserved by homomorphisms is in fact undecidable~\cite{bodirsky2021logics}.

In the present section, we introduce the CSP hammer, a certain syntactic transformation $\Phi \mapsto \Phi^{\mathrm{CSP}}$ for $\forall$-restricted negative SO sentences.
The purpose of the CSP hammer is to turn every $\forall$-restricted negative SO sentence into one that defines a CSP and whose model-checking is polynomial-time equivalent to the original when restricted to a particular subset of instances. 
To this end, we first introduce the \emph{restriction} transformation on  FO sentences.
Let $\Phi$ be a FO $\tau$-sentence, and $U$ a fresh unary relation symbol not in $\tau$.
Set $$\Phi_{|U} \coloneqq  \big( \forall z \ldotp \neg U(z)\big) \vee \Phi',$$  where $\Phi'$ is obtained from $\Phi$ by the following syntactic replacement.
We replace every atomic formula $\psi(x_1,\dots,x_n)$ in $\Phi$ with the formula 
\begin{align}
 \big(\bigvee\nolimits_{i\in I_{\forall}} \neg U(x_i)\big) \vee   \Big(\bigwedge\nolimits_{i\in I_{\exists}}U(x_i)\land \psi(x_1,\dots, x_n) \Big), \label{eq:syntactic_restriction}
\end{align}
where $I_{\exists}$ and $I_{\forall}$ consist of the indices $i \in [n]$ such that $x_i$ is existentially and universally quantified in $\Phi$, respectively.
The purpose of $ \Phi_{|U}$ is captured by the next observation.

\begin{observation}\label{obs:restriction-to-U} Let $\struct{A}$ be an arbitrary $(\tau \cup \{U\})$-structure.
Then $\struct{A}$ satisfies $\Phi_{|U}$ if and only if the substructure of $\struct{A}$ on $U^{\struct{A}}$ satisfies $\Phi$. \hfill \qedsymbol
\end{observation}

Now, let $\Phi \coloneqq\mathrm{Q_1}S_1\dots \mathrm{Q_n}S_n\ldotp\Psi$ be a $\forall$-restricted negative SO $\tau$-sentence such that $\Psi$ is a FO $(\tau\cup \sigma)$-sentence for $\sigma\coloneqq \{S_1,\dots,S_n\}$.
Moreover, let $U$ and $N$ be fresh symbols of arities $1$ and $2$, respectively.
Fix $k$ to be the first index $i\in [n]$ such that $\mathrm{Q_i} = \forall$;  if no such index exists, set $k = n+1$. 
We set  
\begin{align}
    \Phi^{\mathrm{CSP}} \coloneqq   \mathrm{Q_1}S_1\dots \mathrm{Q_k}S_k \forall U \mathrm{Q_{k+1}}S_{k+1}\dots \mathrm{Q_n}S_n  \left(\Psi^{N}_{\neg \mathrm{edge}} \lor (\Psi^{N}_{\neg \mathrm{loop}} \land \Psi)\right)_{|U},
\end{align}
where $\Psi^{N}_{\neg \mathrm{edge}}$ and $\Psi^{N}_{\neg \mathrm{loop}}$ are FO $(\tau\cup \sigma \cup \{N\})$-sentences defined below:
\begin{align}
    \Psi^{N}_{\neg \mathrm{edge}} & \ \coloneqq  \exists x,y \big( (x\neq y) \land \lnot N(x,y)\big) \\
    \Psi^{N}_{\neg \mathrm{loop}} & \ \coloneqq \forall x \ldotp \lnot N(x,x).
\end{align} 
 
It is not hard to see that $\Phi \mapsto \Phi^{\mathrm{CSP}}$ preserves negativity, and $\forall$-restriction\textemdash via the second disjunct in~\eqref{eq:syntactic_restriction}.
It also preserves the alternation-depth for the second-order quantifier prefix of $\Phi$ whenever $\Phi$ contains at least one universally quantified SO variable.
If there is no universal SO variable, then model-checking of $\Phi$ is in NP, which for us is an uninteresting case. 
The next lemma contains the core properties of the CSP hammer.
\begin{restatable}{lemma}{csphammer}   \label{lem:CSP_hammer}    Let $\Phi$  be a $\forall$-restricted negative SO $\tau$-sentence. Then:
    \begin{enumerate} 
        \item \label{item:disjun} $\Phi^{\mathrm{CSP}}$ is preserved under disjoint unions.
        \item \label{item:invhom} $\Phi^{\mathrm{CSP}}$ is preserved under inverse homomorphisms.  
        \item \label{item:evalu} 
        Model-checking for $\Phi$ reduces in polynomial time to  model-checking for $\Phi^{\mathrm{CSP}}$.   
    \end{enumerate}  
    
\end{restatable}
It will be convenient to introduce the following two auxiliary notions.
For a $(\tau\cup \{N\})$-structure $\struct A$, we say that $\struct A$ is:
\begin{itemize}
    \item $N$-\emph{surjective} if the $N$-reduct of $\struct A$ does not satisfy $\Psi^{N}_{\neg \mathrm{edge}}$;
    \item $N$-\emph{injective} if the $N$-reduct of $\struct A$ satisfies $\Psi^{N}_{\neg \mathrm{loop}}$.
\end{itemize}   
\begin{observation}\label{obs:no-loops}
If a $\tau\cup\{N\}$-structure $\struct A$ satisfies $\Phi^{\mathrm{CSP}}$, then it satisfies $\Psi^{N}_{\neg \mathrm{loop}}$. 
\end{observation}
\begin{proof}[Proof of Observation~\ref{obs:no-loops}]
We argue the contrapositive. Let $a\in A$ be such that $\struct A\models N(a,a)$. 
    Consider the following strategy of UP for SO variables: $\ass{U} = \{a\}$ and
    $\ass{S} = \emptyset$ for every universal SO variable $S\neq U$. It readily follows from the 
    definitions that the substructure with domain $\ass{U}$ is $N$-surjective  and  does not satisfy $\Psi_{\neg\mathrm{loop}}$.
\end{proof}

\begin{proof} 
We assume that  $\Phi$ is of the form $ \mathsf{Q_1}S_1\dots \mathsf{Q_n}S_n\ldotp\Psi$, where $\Psi$   a FO $(\tau\cup \sigma)$-sentence for $\sigma\coloneqq \{S_1,\dots,S_n\}$.
We verify the three items in the lemma.
For convenience, we set $$\Psi' \coloneqq \Psi^{N}_{\neg \mathrm{edge}} \lor (\Psi^{N}_{\neg \mathrm{loop}} \land \Psi).$$
For~\eqref{item:disjun}, let $\struct A, \struct B \in \fm(\Phi^{\mathrm{CSP}})$ be arbitrary.
    Consider the following strategy of the EP for evaluating $\Phi^{\mathrm{CSP}}$ on
    $\struct{A+B}$ based on the winning strategies of the EP for evaluating
    $\Phi^{\mathrm{CSP}}$ on $\struct{A}$ and on $\struct B$. For a SO variable $S$ (of arity $k$):
    \begin{itemize}
        \item if $S$ is existential and the EP chooses $\ass{S}  \coloneqq X_A\subseteq A^k$ on
        $\struct{A}$  and $\ass{S} \coloneqq X_B\subseteq B^k$ on $\struct{B}$, then the EP
        chooses $\ass{S} \coloneqq X_A\cup X_B \subseteq (A\cup B)^k$ on $\struct{A+B}$;  
        \item if $S$ is universal and the UP chooses $\ass{S}  \coloneqq Y\subseteq
        (A\cup B)^k$ on $\struct{A+B}$, then we consider the situation where the UP chooses
        $\ass{S} \coloneqq Y\cap A^k$ on $\struct{A}$  and $\ass{S}  \coloneqq Y\cap B^k$ on $\struct{B}$. 
    \end{itemize} 
    Denote by $\struct{A}^+$, $\struct{B}^+$, and $\struct{(A+B)}^+$  the expansions of $\struct{A}$, $\struct{B}$,
    and $\struct{A+B}$ resulting from the above evaluations of the SO variables in $\Phi$, respectively.
    We now consider two possible cases for the evaluation of $\ass{U}$ and show that, in both of them, $\struct{A+B} \models \left(\Psi'\right)_{|U}$. 
    
    \emph{Case 1}: $\ass{U}\subseteq A$ or $\ass{U}\subseteq  B$. 
    Without loss of generality, we assume that  $\ass{U}\subseteq A$.   
    Since $\struct A^+ \models  \left(\Psi'\right)_{|U}$, it follows from Observation~\ref{obs:restriction-to-U} that $\struct{(A+B)}^+ \models  \left(\Psi'\right)_{|U}$.
    
    \emph{Case 2:}  $\ass{U}\cap A \neq \varnothing$ and $\ass{U}\cap B \neq \varnothing$.
    In this case, the substructure of $(\struct{A+B})^+$ with domain $\ass{U}$ satisfies $\Psi^{N}_{\neg \mathrm{edge}}$ and hence also $\left(\Psi'\right)_{|U}$.

   For~\eqref{item:invhom}, observe first that $\Phi^{\mathrm{CSP}}$ is logically equivalent to a $\forall$-restricted negative SO sentence.
   Indeed, the addition of $\Psi^{N}_{\neg \mathrm{edge}}$ and $\Psi^{N}_{\neg \mathrm{loop}}$ preserves negativity, and the application of $(\cdot)_{|U}$ ensures $\forall$-restriction.
   Hence,  we conclude via Lemma~\ref{thm:exist_guarded_negative} that $\Phi^{\mathrm{CSP}}$ is preserved under inverse injective homomorphisms.
   Note that every homomorphism $\struct{A} \rightarrow \struct{B}$ decomposes to a surjective homomorphsism $\struct{A} \rightarrow \struct{A}_q$ and an injective homomorphism $\struct{A}_q \rightarrow \struct{B}$.
   It thus suffices to verify that $\Phi$ is preserved under inverse surjective homomorphisms. 
   
   To this end, let $\struct{A}$ and $\struct{B}$ be two finite $\tau$-structures over the signature of $\Phi^{\mathrm{CSP}}$ such that $\struct{B} \models \Phi$ and there exists a surjective homomorphism $h\colon \struct{A} \rightarrow \struct{B}$. 
   We claim that the EP has the following winning strategy on $\struct{A}$ based on the winning strategy of the EP on $\struct{B}$.
   For a SO variable $S$:
    \begin{itemize}
        \item if $S$ is existentially quantified and the EP chooses $X\subseteq B^k$ for $\ass{S}$ on $\struct{B}$,
        then the EP chooses $h^{-1}(X)\subseteq A^k$ for $\ass{S}$ on $\struct{A}$;  
        \item if $S = U$ and the UP reacts by choosing $Y_U\subseteq A$ for $\ass{U}$, then we consider
        the situation where UP chooses $f(Y_U)\subseteq B$ for $\ass{U}$ in $\struct B$;
        \item if $S$ is universally quantified with $S\neq U$ and UP reacts by
        choosing $Y\subseteq A^k$ for $\ass{S}$ on $\struct{A}$, then we consider the situation where
        UP chooses $h(Y \cap Y_U^k)\subseteq B^k$ for $\ass{S}$ on $\struct{B}$. 
    \end{itemize} 
    Denote by $\smash{\struct{A}^+}$ and $\struct{B}^{\smash{+}}$ the expansions of $\struct{A}$ and $\struct{B}$ resulting from the above evaluations of the SO variables in $\Phi$.
    We show that $\struct A^+ \models  \left(\Psi'\right)_{|U}$.
 
 Note that, if $\ass{U}$ does not induce a $N$-surjective structure in $\struct A$, then $\struct A^+\models (\Psi^{N}_{\neg \mathrm{edge}})_{|U}$ and hence also $\struct A^+\models  \left(\Psi'\right)_{|U}$. 
 Hence, we assume that $\ass{U}$ induces a $N$-surjective structure in $\struct{A}$. 
 Note that, every model of $\Phi^{\mathrm{CSP}}$ is a model $\Psi^N_{\neg \mathrm{loop}}$
 (Observation~\ref{obs:no-loops}), and so, every model of $\Phi^{\mathrm{CSP}}$ is $N$-injective. 
 Since the substructure of $\struct{A}$ on $\ass{U}$ is $N$-surjective and $h$ is a homomorphism, we have that $h$ must be injective on $U^{\struct{A}^+}$.
 In sum, we have that $h$ is a bijection on $U^{\struct{A}^+}$.
 In what follows, it will be convenient to identify the sets $U^{\struct A^+}$ and $U^{\struct B^+}$, and simply denote them by $\ass{U}$. 
 Since $\struct B^+\models \left(\Psi'\right)_{|U}$, the EP has a winning strategy to show that the substructure of $\struct B^+$ with domain $\ass{U}$ satisfies $\Psi'$. 
 Now, the EP can mimic this strategy when evaluating $\Psi'$ on the substructure of $\struct A^+$ induced by $\ass{U}$. 
 The fact that this is a winning strategy follows from the following two facts:
 \begin{itemize}
     \item atomic $(\tau\cup\{N\})$-formulas cannot appear positively in $\Psi'$ and, since $h$ is a homomorphism, for every $R\in\tau\cup\{N\}$,  $R^{\struct{A^+}}\cap \ass{U}^k \subseteq R^{\struct{B^+}}\cap \ass{U}^k$ 
     \item by the strategy of EP for SO variables  and the identification $U^{\struct{A}^+} = \ass{U} = U^{\struct{B^+}}$, for every SO variable $S$ the equality $S^{\struct{A}^+}\cap \ass{U}^k = S^{\struct{B}^+}\cap \ass{U}^k$
    holds.
 \end{itemize}
  Since the EP has a winning strategy for when evaluating $\Psi'$ on the substructure of $\struct A^+$ induced by $\ass{U}$, we conclude that $\struct A^+\models \left(\Psi'\right)_{|U}$. 

   For~\eqref{item:evalu}, we show that, for every $\tau$-structure $\struct A$, we have that $$\struct A\models \Phi \quad \text{if and only if} \quad  (\struct A,  {\neq})\models \Phi^{\mathrm{CSP}}.$$
   First, observe that one direction is trivial.
   Indeed, if $(\struct A,  {\neq})\models \Phi^{\mathrm{CSP}}$, then the EP simply projects the winning strategy onto evaluating $\Phi^{\mathrm{CSP}}$ on $\struct{A}$
   for the case where $\ass{U}=A$.
   Now suppose that $\struct A\models \Phi$.
   Since $\Phi$ is $\forall$-restricted, it is preserved under substructures due to Lemma~\ref{thm:exist_guarded_negative}.
   This means that the EP has a winning strategy for evaluating $\Phi$ on any substructure of $\struct{A}$.
   But note that the notion of $\forall$-restriction implies an even stronger property.
   Namely, it follows from Definition~\ref{def:forallguarding} that the EP must have $\mathsf{Q}_1,\dots, \mathsf{Q}_k$-uniform winning strategies for evaluating $\Phi$ on the substructures of $\struct{A}$, in the sense that the winning strategies coincide for the evaluation of the first block of existential SO variables $S_1,\dots, S_k$.
   The reason is that the EP cannot possibly know which substructure the UP will restrict the game to using the universal variables until the first universal SO is played.
   From this observation, it follows immediately that the EP can extend the $\mathsf{Q}_1,\dots, \mathsf{Q}_k$-uniform winning strategy for evaluating $\Phi$ on the substructures of $\struct A$ to a winning strategy for evaluating $\Phi^{\mathrm{CSP}}$ on $(\struct A,  {\neq})$; the choice of $\ass{U}$ done by the UP restricts the substructures of $\struct A$ the UP might want to play on.
   
   Formally, we have the following two cases for a SO variable $S$.
   If $S \in \{S_1,\dots, S_k\}$ and the EP chooses $X\subseteq A^k$ for $\ass{S}$ on $\struct{A}$, then the EP chooses the same set  $X\subseteq A^k$ for $\ass{S}$ on $(\struct A,  {\neq})$.
   Otherwise $U$ precedes $S$. Then:  
        \begin{itemize}
                \item if $S$ is existentially quantified and the EP chooses $\ass{S} =X\subseteq A^k$
                for evaluating $\Phi$ on $\struct A$, then EP chooses $\ass{S} =X\subseteq A^k$ for evaluating $\Phi^{\mathrm{CSP}}$ on $(\struct A, {\neq})$;
                \item if $S\neq U$ is universally quantified and the UP reacts by choosing 
                $\ass{S} =X\subseteq A^k$ for evaluating $\Phi^{\mathrm{CSP}}$ on $(\struct A, {\neq})$,
                then we consider the case when UP chooses $\ass{S} = X \cap \ass{U}^k\subseteq A^k$ for evaluating
                $\Phi$ on $\struct A$. \qedhere
        \end{itemize} 
\end{proof}


\section{From finite-domain QCSPs to $\omega$-categorical CSPs} \label{sect:odd_part}

In the present section, we give a general construction of $\omega$-categorical CSPs from bounded alternation finite-domain QCSPs.
%
Let $\struct{B}$ be a finite structure with the relational signature $\tau$ and domain $B=\{b_1,\dots, b_{\ell}\}$.
Also, fix an FO quantifier-prefix $\mathrm{Q}_1\bar{x}_1\cdots \mathrm{Q}_n\bar{x}_n$ for some $n\geq 1$, where $\mathrm{Q}_i\in \{\exists,\forall\}$ for every $i\in [n]$.
We define a new signature $\tau'$ consisting of the symbols from $\tau$, and, for every $i\in [n]$,
the unary symbol $\temp{E}_i$ or $\temp{A}_i$, depending on whether $\mathrm{Q}_i$ equals $\exists$ or $\forall$.
We shall refer to the $i$th symbol $\temp{E}_i$ or $\temp{A}_i$ by $\temp{Q}_i$.
Consider the MSO $\tau'$-sentence 
$$\Phi_{\struct{B}} \coloneqq \mathrm{Q}_1 \temp{Q}_{1,b_1} \cdots   \mathrm{Q}_1\temp{Q}_{1,b_{\ell}}   
 \cdots \mathrm{Q}_n \temp{Q}_{n,b_1} \cdots  \mathrm{Q}_n\temp{Q}_{n,b_{\ell}} \left( \Psi_\forall \lor  (\Psi_\exists \land \Psi_{\struct{B}}) \right),$$
where $\Psi_\forall$, $\Psi_{\exists}$,  and $\Psi_{\struct{B}}$ are FO formulas defined below.
Intuitively,  $\temp{Q}_{k,b}$ for $k\in [n]$ and $b\in B$ represents the set of all variables $v$ in an instance of $\mathrm{Q}_1\cdots \mathrm{Q}_n$$\QCSP(\struct{B})$ for which the player controlling $\mathrm{Q}_k$ has chosen $\ass{x}\coloneqq b$. 
We now give definitions of $\Psi_\forall$,  $\Psi_{\exists}$,  and $\Psi_{\struct{B}}$.

First, $\Psi_{\forall}$ states that the UP did not  play by the rules:
   \begin{align} 
    \Psi_\forall\coloneqq \bigwedge\nolimits_{i\in [n]:\, \temp{Q}_i=\temp{A}_i}  \bigvee\nolimits_{b\in B}\bigvee\nolimits_{c\in B\setminus \{b\}}  \exists x \big(\lnot A_i(x) \land   A_{i,b}(x) \big) \vee \exists x \big( A_{i,b}(x)\land A_{i,c}(x)\big). \label{eq:universal_chooses_correctly}
    \end{align} 
     
%
%
The sentence $\Psi_{\exists}$ states that the EP   plays by the rules:
 \begin{align}
    \Psi_\exists  \coloneqq \bigwedge\nolimits_{i\in [n]:\,\temp{Q}_i=\temp{E}_i}  \forall x   \left(\neg \temp{E}_i(x) \vee \bigvee\nolimits_{b\in B} \temp{E}_{i,b}(x)\right).
    \label{eq:existential_chooses_correctly}
   \end{align}  

Finally, $\Psi_{\struct{B}}$ states that either the UP has cheated---by violating the converse of eq.~\eqref{eq:universal_chooses_correctly}---or the $\mathrm{Q}_i$s and $Q_{i,b}$s are compatible with the relations of $\struct{B}$:
        \begin{align} 
     \Psi_{\struct{B}} = \bigwedge\nolimits_{R\in \tau \text{ with arity }m}  \bigwedge\nolimits_{(i_1,\dots, i_m)\in  
        [n]^m} \forall \bar{x} \Big(     \psi_{\tau'}   \vee \psi_{\forall}  \vee \psi_{\exists}    \Big),
  \label{eq:main_part}
    \end{align} 
    where $\psi_{\tau'}$, $\psi_{\forall}$, and $\psi_{\exists}$ are shortcuts for the following formulas:
    \begin{align*}
        \psi_{\tau'}   &  \coloneqq \neg R(x_1,\dots, x_m)  \vee \bigvee\nolimits_{k\in [m]} \lnot Q_{i_k}(x_k), \\
         \psi_{\forall} & \coloneqq  \bigvee\nolimits_{k\in [m]: \, Q_{i_k} = A_{i_k}}\bigwedge\nolimits_{b\in B} \neg A_{i_k,b}(x_k) , \\
         \psi_{\exists}  & \coloneqq   \bigvee\nolimits_{(t_1,\dots,t_m)\in R^{\struct{B}}}  \bigwedge\nolimits_{k\in [m]} Q_{i_k,t_k}(x_k).  
    \end{align*}


\begin{restatable}{lemma}{qcsp}      \label{lem:QCSP}  
\begin{enumerate}
    \item \label{item:preserv} $\Phi_{\struct{B}}$ is $\forall$-restricted and negative.
    \item \label{item:reduct} $\mathrm{Q}_1\cdots \mathrm{Q}_n$$\QCSP(\struct{B})$  reduces in polynomial-time to model-checking for $\Phi_{\struct{B}}$. 
\end{enumerate}
 
\end{restatable}  
\begin{proof} 
For~\eqref{item:preserv}, we use the fact that all involved notions are closed under internal conjunctions and disjunctions. 
For negativity, it is enough to check that every $\tau'$-atom other than equality appears under the scope of an odd number of negations.
We continue with $\forall$-restriction.
That $\Psi_{\forall}$ is $\forall$-restricted can be checked directly using Definition~\ref{def:forallguarding}. 
That $\Psi_{\exists}$ and $\Psi_{\struct{B}}$ are $\forall$-restricted follows from the fact that the sentences do not contain any existentially quantified FO variables. 

For~\eqref{item:reduct}, we use the following reduction.
Given an instance $\Phi\coloneqq \mathsf{Q}_1 \bar{x}_1 \cdots \mathsf{Q}_n \bar{x}_{n}\ldotp \phi$ of $\QCSP(\struct{B})$, we obtain the instance 
$ 
\Psi \coloneqq \exists \bar{x}_1 \cdots \exists \bar{x}_{n}\ldotp \psi
$ 
of $\CSP(\struct{D})$, where $\psi$ is the conjunction of the following atomic $\tau'$-formulas:
\begin{itemize}  
    \item $\temp{Q}_i(v)$ if and only if $v$ appears in $\bar{x}_i$; 
    \item $R(v_1,\dots, v_n)$ for $R\in \tau$ if and only if this atom appears in $\phi$.
\end{itemize}
Let $\struct{A}$ be the representation of $\Psi$ as a relational $\tau'$-structure: the domain  $A$ of $\struct{A}$ consists of the variables in $\Psi$, and $\bar{a}\in R^{\struct{A}}$ if and only if $R(\bar{a})$ is an atomic formula in $\Psi$.

``$\Rightarrow$'' Suppose that $\Phi$ is a YES-instance of $\QCSP(\struct{B})$.
Before defining a strategy for the EP, notice that $\Psi_\forall$ says that if the evaluation
of universal SO variables do not satisfy the inclusion  $\ass{A_{i,b}}\subseteq A_i$, then
EP can trivially win the FO part. So we assume that $\ass{A_{i,b}}\subseteq A_i$, and
consider the following strategy of the EP evaluating the SO variables of $\Phi_{\struct{B}}$ on
$\struct{A}$ based on the winning strategy of the EP on $\Phi$:
    \begin{itemize}
        \item the EP evaluates SO variables on $\struct{A}$: $\ass{\temp{E}_{i,b}} \coloneqq \{v \in \bar{x}_i \mid \ass{v} = b \text{ on } \Phi \}$;  
        \item if the UP reacts by evaluating $\ass{\temp{A}_{i,b}} = X\subseteq A$, then the EP considers 
        the move of the UP on the QCSP instance where: $\ass{v} \coloneqq b$ if and only if
        $v\in \ass{\temp{A}_{i,b}}$ on $\struct{A}$---if this is not a valid move for a
        variable $b$, i.e.,  $v$ that gets assigned no value $b\in B$ or two different values
        $b\neq b'$, then EP chooses an arbitrary value $\ass{v} \in B$.
    \end{itemize}
  Let $\smash{\struct{A}^+}$ be an expansion of $\struct{A}$ resulting from the above evaluations of the SO variables in $\Phi_{\struct{B}}$.
  We show that $\smash{\struct{A}^+}$ satisfies the FO part of $\Phi_{\struct{B}}$, and it suffices
  to show that it satisfies $\Psi_\exists \land \Psi_{\struct{B}}$.
  Clearly, we have $\smash{\struct{A}^+}\models  \Psi_{\exists}$ because the EP assigns a value to every existential variable on $\Phi$.
    We must show that $\smash{\struct{A}^+}\models \Psi_{\struct{B}}$.
    Whenever $\psi_{\tau'}$ 
    is satisfied, we have the following two cases.
    Either the UP has not assigned values to the $\temp{A}_{i,b}$s in a way that would correctly represent an assignment of values to the universal variables in $\Phi$, in which case $\psi_{\forall}$ 
    holds, or $\psi_{\exists}$ holds 
    because the strategy of the EP on $\struct{A}$ reflects
    a winning strategy of the EP on $\Phi$. 
  
``$\Leftarrow$'' Suppose that $\Phi$ is a NO-instance of $\QCSP(\struct{B})$.
Consider the following strategy of the UP for evaluating $\Phi_{\struct{B}}$ on $\struct{A}$ based on the winning strategy of the UP on $\Phi$:
    \begin{itemize}
        \item the UP evaluates SO variables on $\struct{A}$:   $\ass{\temp{A}_{i,b}} \coloneqq \{v \in \bar{x}_i \mid \ass{v} = b \text{ on } \Phi \}$;
        \item the EP evaluates FO variables on $\Phi$: $\ass{v} \coloneqq b$ if and only if $v\in \ass{\temp{E}_{i,b}}$ on $\struct{A}$.
    \end{itemize}
 Let $\smash{\struct{A}^+}$ be an expansion of $\struct{A}$ resulting from the above evaluations of the SO variables in $\Phi_{\struct{B}}$.
 We show that $\smash{\struct{A}^+}$ does not satisfy the FO part of $\Phi_{\struct{B}}$.

 Either eq.~\eqref{eq:existential_chooses_correctly} does not hold, in which case the EP loses immediately, or the assignment of values to the existentially quantified SO variables in $\Phi_{\struct{B}}$ on $\struct{A}$ correctly represents an assignment of values to the existentially quantified variables in $\Phi$; we assume the latter.
 Regarding $\Psi_{\forall}$,  eq.~\eqref{eq:universal_chooses_correctly} holds by the strategy of the UP on $\struct{A}$; hence, $\smash{\struct{A}^+} \models \Psi_{\forall}$. 
 We show that $\smash{\struct{A}^+} \centernot{\models} \Psi_{\struct{B}}$. 
 By the strategy of the UP on $\struct{A}$, if $\psi_{\tau'}$ does not hold, then $ \psi_{\forall}$ does not hold either.
 %
 But then $\psi_{\exists}$ cannot be satisfied uniformly across all relations because the strategy of the UP on $\struct{A}$ reflects a winning strategy of the UP on $\Phi$. 
\end{proof}  

\begin{proof}[Proof of Theorem~\ref{thm:qcsp_to_csp}]  
Consider the MSO sentence $\Phi_{\struct{B}}$.
By Lemma~\ref{lem:QCSP}\eqref{item:preserv}, we have that $\Phi_{\struct{B}}$ is $\forall$-restricted and negative.
Hence, the prerequisites of Lemma~\ref{lem:CSP_hammer} are satisfied for $\Phi_{\struct{B}}$.
By items~\ref{item:disjun} and~\ref{item:invhom} of Lemma~\ref{lem:CSP_hammer}, $\Phi_{\struct{B}}^{\mathrm{CSP}}$ is preserved by disjoint unions and inverse homomorphisms.
By construction, $\Phi_{\struct{B}}^{\mathrm{CSP}}$ is a $\mathrm{Q}_1\cdots \mathrm{Q}_n$MSO sentence.
It follows from Theorem~\ref{theorem:omega_cat_CSPs} that there 
exists an $\omega$-categorical structure $\struct{D}$ with $\CSP(\struct{D})=\fm(\Phi^{\mathrm{CSP}}_{\struct{B}})$.
By Lemma~\ref{lem:QCSP}\eqref{item:reduct}, $\mathrm{Q}_1\cdots \mathrm{Q}_n$$\QCSP(\struct{B})$  reduces in polynomial-time
to model-checking for $\Phi_{\struct{B}}$. 
It follows from item~\ref{item:evalu} of Lemma~\ref{lem:CSP_hammer} that model-checking for $\Phi_{\struct{B}}$ reduces in polynomial-time to model-checking for $\Phi_{\struct{B}}^{\mathrm{CSP}}$. 
\end{proof}  
\begin{proof}[Proof of Theorem~\ref{thm:odd_part}]  As already mentioned, Theorem~\ref{thm:odd_part} follows from Theorem~\ref{thm:qcsp_to_csp}, and from the fact
that there are $(\exists\forall)^n\exists$- and $\forall(\exists\forall)^n\exists$-QCSPs
complete for the $\Sigma^{\PO}_{2n+1}$-level and the $\Pi^{\PO}_{2n}$-level of the polynomial hierarchy ($n\geq 1$);
e.g., $(\exists\forall)^n\exists$- and $\forall(\exists\forall)^n\exists$-3CNF.
\end{proof}

\section{$\omega$-categorical CSPs complete for $\Sigma^{\mathrm{P}}_{2n}$ and $\Pi_{2n+1}^\PO$}
\label{sect:even_part}

In the present section, we give a construction of $\Sigma^{\mathrm{P}}_{2n}$-complete
$\omega$-categorical CSPs from the complement of $(\forall\exists)^n$3-CNF. 
There is an analogous construction of $\Pi^{\mathrm{P}}_{2n+1}$-complete
$\omega$-categorical CSPs from the complement of $\exists(\forall\exists)^n$3-CNF; for the sake of simplicity, we choose to only provide details for the case of $(\forall\exists)^n$3-CNF.
%
  
Fix $n\geq 1$. 
The strategy is to construct an $\exists$-guarded positive MSO $\tau$-sentence $\Phi_{\ast}$ such that the SO quantifier prefix of $\Phi_{\ast}$ is $(\forall\exists)^n$ and $(\forall\exists)^n$3-CNF reduces in polynomial-time to model-checking for $\Phi_{\ast}$.
Then we hit $\neg \Phi_{\ast}$ with the CSP hammer.
%

We define the signature $\tau$ consisting of the binary symbols $R$, $\smash{\overline{R}}$,
and $\succ$, and the unary symbols $S$, $T$,  $E_1,\dots, E_n$, and  $A_1,\dots, A_n$.
Now, we set
\[
\Phi_{\ast} \coloneqq \forall\temp{A}_{1,0}\exists \temp{E}_{1,1}\ldots \forall\temp{A}_{n,0}\exists \temp{E}_{n,1}\exists V \big(\lnot\Psi_\forall  \lor (\Psi_\exists \land    \Psi_{|V})\big),
\]
where $V$ is a fresh unary symbol,  and $\Psi_{\forall}$, $\Psi_{\exists}$, $\Psi$ are FO formulas defined below.
At a high-level, $\Psi_\forall$ and $\Psi_\exists$ are sets of FO constraints that
UP and the EP must follow in their SO evaluations strategies. Following this intuition, the FO part of $\Phi_\ast$ says that, whenever the UP plays by the rules $\Psi_\forall$, then EP must play by the rules $\Psi_\exists$ and satisfy $\Psi$ in the substructure with domain $\ass{V}$. 

Formally, $\Psi_\forall$ it is the conjunction of the following sentences:
\begin{align}
 \forall x \big(\neg \temp{A}_{k,0}(x) \vee \temp{A}_{\ell,0}(x) \big) & \text{ for all }  k,\ell\in [n] \text{ with } k < \ell, \label{eq:the UP_consistent_upwards_across_prefix}     \\ 
 \forall x \big(\neg \temp{A}_{\ell,0} (x) \vee \neg\temp{A}_k (x) \vee \temp{A}_{k,0}(x)\big) & \text{ for all }   k,\ell\in [n] \text{ with } k < \ell, \label{eq:the UP_consistent_downwards_across_prefix} \\ 
  \forall x \big( \neg \temp{A}_{k,0} (x) \vee \lnot\temp{A}_{\ell}(x) \big) & \text{ for all }  k,\ell\in [n] \text{ with } k < \ell, \label{eq:the UP_does_not_choose_to_soon} \\ 
  \forall x \big( \neg \temp{A}_{n,0}(x)\vee \lnot\temp{E}_k(x) \big) & \text{ for each } k\in[n], \label{eq:the UP_does_not_choose_existentially} \\ 
   \forall x,y \big( \neg\temp{A}_{n,0}(x) \vee \neg \temp{A}_{n,0}(y) \vee \neg \smash{\overline{R}}(x,y) \big).& \label{eq:the UP_consistent_with_negated_pairs}
\end{align}

Now,  $\Psi_{\exists}$ is the  conjunction of the following formulas:
\begin{align}
   \forall x &\bigwedge\nolimits_{1\leq k<\ell\leq n}  \left(\lnot E_{k,1}(x)\lor E_{\ell,1}(x)\right), \label{eq:the EP_consistent_upwards_across_prefix} \\
   \forall y &\bigwedge\nolimits_{k\in[n]} \big( \lnot E_{k,1}(y) \lor \bigvee\nolimits_{j\in[k]}E_j(y)\big), \label{eq:the EP_only_chooses_when_its_time} \\
   \forall w,z &  \left( \lnot E_{n,1}(w) \lor \lnot E_{n,1}(z) \lor R(w,z)\right).  \label{eq:the EP_consistent_with_negated_pairs} 
\end{align} 


%

Finally, $\Psi$ is the conjunction of the following formulas: 
\begin{align} 
\forall x,y\ldotp R(x,y) \label{eq:satisf}\\
  \exists x,y \big(S(x) \land T(y)\big),  \label{eq:start}  
\\ 
\forall x \big(\lnot \temp{A}_{n,0}(x) \vee  \smash{\bigvee\nolimits_{k\in [n]}\temp{E}_{k,1}(x)}  \big), \label{eq:consist} 
\\
  \forall x,y \, \exists a,b \big((x=y) \vee  
 \big( (x\prec a) \wedge (b\prec y)  \big)\big).
  \label{eq:succ}   
\end{align} 
 
\begin{restatable}{lemma}{sat}      \label{lem:SAT}
\begin{enumerate}
    \item \label{item:preserv2} $\Phi_{\ast}$ is logically equivalent to an $\exists$-guarded positive sentence.
    \item \label{item:reduct2} $(\forall\exists)^n$3-CNF reduces in polynomial-time to model-checking for $ \Phi_{\ast}$. 
\end{enumerate}
 
\end{restatable}  
\begin{proof} For~\eqref{item:preserv2}, we use the dualities from Lemma~\ref{lem:dualizing_positivity} and Lemma~\ref{lem:dualizing_exists_guarding}, as well as 
the fact that all involved notions are closed under internal conjunctions and disjunctions. 
For positivity, it is enough to check that every $\tau$-atom appears under the scope of an even number of negations.
Indeed, every $\tau$-atom appears negated in $\Psi_{\forall}$, 
and non-negated otherwise.
We continue with $\exists$-guarding.
Regarding $\Psi_\forall$, a simple application of De Morgan's laws shows that there is sentence $\overline{\Phi_{\forall}}$ logically equivalent to $\neg \Psi_\forall$ and such that
$$
\forall\temp{A}_{1,0}\exists \temp{E}_{1,1}\ldots \forall\temp{A}_{n,0}\exists \temp{E}_{n,1}\exists V\ldotp \overline{\Phi_{\forall}}
$$
is $\exists$-guarded, because it does not contain any universally quantified FO variables.
Regarding $\Psi_\exists$, it follows immediately from the definition of $\exists$-guarding that
$$\forall\temp{A}_{1,0}\exists \temp{E}_{1,1}\ldots \forall\temp{A}_{n,0}\exists \temp{E}_{n,1}\exists V\ldotp \Psi_\exists
$$
is $\exists$-guarded. Finally, a simple application of the distributive law of propositional logic shows that there is a sentence $\Psi_{V}$ logically equivalent to $\Psi_{|V}$ and such that
$$
\forall\temp{A}_{1,0}\exists \temp{E}_{1,1}\ldots \forall\temp{A}_{n,0}\exists \temp{E}_{n,1}\exists V\ldotp  \Psi_{V}
$$
is CNF-normalized and $\exists$-guarded. Hence, $\Phi_{\ast}$ is $\exists$-guarded up to logical equivalence.

For~\eqref{item:reduct2}, we consider the following reduction.
The idea is to transform a given instance $\Phi$ of $(\forall\exists)^n$3-CNF into a finite $\tau$-structure $\struct{A}$ where 
$\temp{A}_{k}(v)$ holds if $v$ is a literal of the form $x$ or $\neg x$ for some $x$ universally quantified at the $k$-th position in the quantifier prefix of $\Phi$.
Then $\temp{A}_{\ell,0}(v)$ for $\ell\geq k$ holds in an expansion of $\struct{A}$ satisfying the FO part of $\Phi_{\ast}$ if the UP has chosen $\ass{v}\coloneqq 0$; $\neg \temp{A}_{\ell,0}(v)$ indicates that the UP has chosen $\ass{v}\coloneqq 1$.
The intuition for $\temp{E}_{k}(v)$ and $\temp{E}_{k,1}(v)$ is identical, except that we substitute the EP for the UP and the value $1$ for the value $0$. 
Note that the universal and the existential SO variables in $\Phi_{\ast}$ are treated differently, this is necessary if we want $\Phi_{\ast} $ to be closed under homomorphisms.
In structures which correctly encode instances of $(\forall\exists)^n$3-CNF,  this apparently different treatment will collapse to the same semantic condition. 
 
    Fix an instance of $(\forall\exists)^n$3-CNF
    \[ 
        \Phi\coloneqq\forall \bar x_1\exists \bar y_1 \dots \forall \bar x_n \exists \bar y_n \big(\psi_1\land\dots \land \psi_{m}\big),
    \] 
    where each $\psi_i$ is a disjunction of literals. 
    We denote by $\Lambda$ the set of all literals that appear in $\Phi$, i.e., 
    each variable $v$ appears in $\Lambda$ both positively as $v$ and negatively as $\neg v$.
    For convenience, we adopt the following terminology.
    For $\lambda\in \Lambda$,  we write:
    \begin{itemize}
        \item $\lambda \in \psi_i$ ($i\in [m]$) if $\lambda$ appears as a literal in the clause $\psi_i$; 
        \item $\lambda \in \bar{x}_k$ or $\lambda \in \bar{y}_k$ ($k\in [n]$) if $\lambda$ equals $v$ or $\neg v$ and $v\in  \bar{x}_k$ or $v\in  \bar{y}_k$.
    \end{itemize}

    We now construct a $\tau$-structure $\struct{A}$ from $\Phi$.
    The domain is  
    \[
        A\coloneqq\{(\lambda,i)\in \Lambda\times [m]  \mid   \lambda \in  \psi_i\},
    \] 
    and the relation symbols in $\tau$ interpret as follows:
    \begin{align} 
         S^{\struct{A}}\coloneqq \ & \Lambda\times\{1\} \nonumber \\
         T^{\struct{A}}\coloneqq \ & \Lambda\times \{m\} \nonumber \\
         {\succ}^{\struct{A}} \coloneqq \ & \{ \big((\lambda,i),(\lambda',j)\big) \in ( \Lambda\times [m])^2 \mid j = i+1 \} \nonumber \\ 
         R^{\struct{A}} \coloneqq \ &  \{ \big((\lambda,i),(\lambda',j)\big) \in ( \Lambda\times [m])^2 \mid \lambda \text{ is not the negation of } \lambda' 
         \} \label{eq:nonnegation} \\ 
         \smash{\overline{R}}^{\struct{A}} \coloneqq \ &  \{ \big((\lambda,i),(\lambda',j)\big) \in ( \Lambda\times [m])^2 \mid \lambda \text{ is the negation of } \lambda'  \} \label{eq:negation} \\
         \temp{A}_k^{\struct{A}} \coloneqq \ & \{(\lambda,i) \in \Lambda\times [m] \mid \lambda \in \psi_i \text{ and }  \lambda \in \bar x_k    \}  \qquad  (k\in [n]) \label{eq:forall} \\
         \temp{E}_k^{\struct{A}} \coloneqq \ & \{(\lambda,i) \in \Lambda\times [m] \mid \lambda \in \psi_i \text{ and }  \lambda \in \bar y_k  \} \qquad  (k\in [n])   \label{eq:exist}
    \end{align}
    Now we show that $\struct{A}\models \Phi_{\ast}$ if and only if $\Phi$ is satisfiable.
    To this end, consider the mapping $\xi$ defined below. 
    The domain of $\xi$ is the set $\mathcal{A}$ of all Boolean assignments $\bar{a}\in \{0,1\}^\Lambda$ that are consistent for variables and their negations, and its range is the set $\mathcal{E}$ of all
    $\big (\{\temp{A}_{1,0}, \temp{E}_{1,1},\dots, \temp{A}_{n,0}, \temp{E}_{n,1}\}\cup \tau\big)$-expansions of $\struct{A}$. 
    
    Given $\bar{a} \in \mathcal{A}$, we define $\xi(\bar{a})$ as the expansion $\struct{A}^+$ where:
    \begin{align*}
        \temp{A}_{\ell,0}^{\struct{A}^+} \coloneqq \ &  \{(\lambda,i) \in \Lambda\times [m] \mid \bar{a}(\lambda)=0 \text{, } \lambda \in \psi_i \text{, and } \lambda \in \bar x_{k} \text{ for some } k \in [\ell]   \} & (\ell\in [n]) \\
        \temp{E}_{\ell,1}^{\struct{A}^+} \coloneqq \ &  \{(\lambda,i) \in \Lambda\times [m] \mid \bar{a}(\lambda)=1 \text{, } \lambda \in \psi_i \text{, and } \lambda \in \bar y_{k} \text{ for some } k \in [\ell]   \} & (\ell\in [n])  
    \end{align*}

    Observe that $f$ is an injective mapping. 
    We claim that $\xi(\bar{a})\models \Psi_\forall \land \Psi_\exists$  for each $\bar{a}\in \mathcal{A}$.

    \medskip First, $\xi(\bar{a})$ satisfies each of the conjunct of  $\Psi_\forall$:
    \begin{itemize}
        \item eq.~\eqref{eq:the UP_consistent_upwards_across_prefix} because $\temp{A}_{k,0}(x)$ extends upwards from $k$ to $n$;
        \item eq.~\eqref{eq:the UP_consistent_downwards_across_prefix} because $\temp{A}_{\ell,0}(x)$ properly extends downwards to $\temp{A}_{k}(x)$; 
        \item eq.~\eqref{eq:the UP_does_not_choose_to_soon} because $\temp{A}_{k,0}(x)$ only holds if $\temp{A}_{k}(x)$ does not hold;
        \item eq.~\eqref{eq:the UP_does_not_choose_existentially} because $\temp{A}_{k,0}(x)$ never holds for literals of existential variables;
        \item eq.~\eqref{eq:the UP_consistent_with_negated_pairs} because $\temp{A}_{k,0}(x)$ never holds both for a variable and its negation. 
    \end{itemize}

    Next, $\xi(\bar{a})$ satisfies all conjuncts in $\Psi_\exists$:
    \begin{itemize}
        \item eq.~\eqref{eq:the EP_consistent_upwards_across_prefix} because $\temp{E}_{k,1}(x)$ extends upwards to $n$;
        \item eq.~\eqref{eq:the EP_only_chooses_when_its_time} because $\temp{E}_{k,1}(x)$ only holds for if $\temp{E}_{k}(x)$ holds as well; 
        \item eq.~\eqref{eq:the EP_consistent_with_negated_pairs} because $\temp{E}_{k,1}^{\struct{A}^+}$ never holds both for a variable and its negation.  
    \end{itemize}

    Now recall that $\struct{A}$ encodes $\Phi$ in a very specific way\textemdash in accordance with the definitions in eqs.~\eqref{eq:nonnegation},~\eqref{eq:negation},~\eqref{eq:forall}, and~\eqref{eq:exist}. 
    It is easy to verify that, for this reason, we also get the converse:   if $\struct{A}^+\in \mathcal{E}$ satisfies $\Psi_\forall \land\Psi_\exists$, then $\struct{A}^+ \in \xi(\mathcal{A})$.
    In sum, we have that $$\xi(\mathcal{A}) = \{ \struct{A}^+\in \mathcal{E} \mid \struct{A}^+ \models \Psi_\forall \land\Psi_\exists \}. $$

     Now, if $\bar{a} \in \mathcal{A}$ satisfies the quantifier-free part of $\Phi$, then for each clause $\psi_i$ there is a literal $\lambda_i$ whose value is $1$ under $\bar{a}$. 
    We claim that then the expansion  $$(\xi(\bar{a}),\{(\lambda,i) \in \Lambda \times [m] \mid \lambda = \lambda_i\})  $$
   (by a relation interpreting $V$)  satisfies all conjuncts in $\Psi$:
    \begin{itemize}
        \item eq.~\eqref{eq:consist} because $V$ selects only those literals that were assigned $1$;
        \item eq.~\eqref{eq:satisf} because $V$ does not contain both a variable and its negation;
        \item eqs.~\eqref{eq:start} and~\eqref{eq:succ} hold because $V$ selects at least one literal in each clause.  
    \end{itemize} 
 
 Conversely, if $V\subseteq A$ is such that $(\xi(\bar{a}),V)\models\Psi_{|V}$, then
    the set of literals $\lambda$ where $(\lambda,i)\in V$ witnesses
    that $\bar{a}$ is a satisfying Boolean assignment of the variables of $\Phi$.

    The mapping $F$ can be naturally extended to partial Boolean assignments of variables
    of $\Phi$. It is straightforward to push and pull winning strategies for the
    the EP from $\Phi$ to winning strategies for the EP for $\Phi_{\ast}$ in
    $\struct{A}$, and vice versa.    
\end{proof}

\begin{proof}[Proof of Theorem~\ref{thm:even_part}] 
Consider the MSO sentence $\Phi_{\ast}$.
By Lemma~\ref{lem:SAT}\eqref{item:preserv2}, $\Phi_{\ast}$ is logically equivalent to an $\exists$-guarded positive sentence.
By Lemma~\ref{lem:dualizing_positivity}, Lemma~\ref{lem:dualizing_exists_guarding}, and Lemma~\ref{thm:exist_guarded_negative}, $\neg \Phi_{\ast}$ is logically equivalent to an $\forall$-restricted negative sentence.
Hence, the prerequisites of Lemma~\ref{lem:CSP_hammer} are satisfied for $\neg \Phi_{\ast}$.
By items~\ref{item:disjun} and~\ref{item:invhom} of Lemma~\ref{lem:CSP_hammer}, we have that $(\neg \Phi_{\ast})^{\mathrm{CSP}}$ is preserved by disjoint unions and inverse homomorphisms.
Note that, by construction, $(\neg \Phi_{\ast})^{\mathrm{CSP}}$ is an $(\exists \forall)^n$-MSO sentence.
It follows from Theorem~\ref{theorem:omega_cat_CSPs} that there 
exists an $\omega$-categorical structure $\struct{D}$ with $\CSP(\struct{D})=\fm((\neg \Phi_{\ast})^{\mathrm{CSP}})$.
By Lemma~\ref{lem:SAT}\eqref{item:reduct2}, $(\forall\exists)^n$3-CNF has a polynomial-time reduction
to model-checking for $\Phi_{\ast}$. 
All our reductions are many-one, and hence it follows immediately that $(\exists\forall)^n$3-DNF has a polynomial time many-one reduction to model-checking for $\neg \Phi_{\ast}$. 
It follows from item~\ref{item:evalu} of Lemma~\ref{lem:CSP_hammer} that model-checking for $\neg \Phi_{\ast}$ reduces in polynomial-time to model-checking for $(\neg \Phi_{\ast})^{\mathrm{CSP}}$. 
We conclude that $\CSP(\struct{D})$ is $\Sigma_{2n}^\PO$-complete.
The $\Pi_{2n+1}^\PO$-complete case can be treated analogously by
adding an existential SO variable $\exists E_{0,1}$ to $\Phi^\ast$ in the outermost
position, and then modifying $\Psi_\exists$ and $\Psi$, accordingly.
\end{proof}


\section{Conclusion and outlook}

We showed that there are $\omega$-categorical CSPs complete for any level of the polynomial hierarchy above \textrm{NP}.
To this end, we developed a new tool, the CSP hammer, for producing MSO sentences closed under inverse homomorphisms and disjoint unions.
By a recent result of Bodirsky, Kn\"{a}uer, and Rudolph~\cite{bodirsky2020datalog}, every such sentence defines an $\omega$-categorical CSP.  
Next, in Section~\ref{sect:odd_part}, we gave a general construction of $\omega$-categorical CSPs from bounded alternation finite-domain QCSPs.
When adjusted with the CSP hammer, this construction is able to produce $\Sigma^{\PO}_{2n+1}$- and $\Pi^{\PO}_{2n}$-complete $\omega$-categorical CSPs for every $n \geq 1$.
One drawback of this method is that it does not always preserve the complexity of the original QCSP; achieving a  fine-grained correspondence between $\omega$-categorical CSPs and bounded alternation finite-domain QCSPs would require turning the CSP hammer into a more robust toolset.
We ask the following question, which seems to be non-trivial.

\medskip
\noindent\textbf{Question 1:} Is it true that for every finite structure $\struct B$ and every quantifier prefix $\mathrm{Q}_1\cdots \mathrm{Q}_n$ there exists an $\omega$-categorical structure
$\struct D$ such that $\mathrm{Q}_1\cdots \mathrm{Q}_n$QCSP($\struct B$) and $\CSP(\struct D)$ are polynomial-time equivalent?
\medskip 

We remark that a positive answer would imply the existence of $\omega$-categorical CSPs complete for some exotic complexity classes in $\PH$ such as DP or $\Theta^{\text{P}}_2$~
\cite{zhuk2022qcsp}.
The obstacle is getting around inverse surjective homomorphisms (see Example~\ref{ex:positivity_sucks}), which our CSP hammer resolves in an adhoc manner.
We ask the same question for the unbounded case.

\medskip
\noindent\textbf{Question 2:} Is it true that for every finite structure $\struct B$ there is an $\omega$-categorical structure
$\struct D$ such that QCSP($\struct B$) and $\CSP(\struct D)$ are polynomial-time equivalent?

\medskip 

In Section~\ref{sect:even_part}, we gave a polynomial time reduction from the complement $(\forall\exists)^n$3-CNF to model checking for MSO sentences which, when combined with the CSP hammer, produces 
$\Sigma^{\PO}_{2n}$- and $\Pi^{\PO}_{2n+1}$-complete $\omega$-categorical CSPs for every $n \geq 1$.
In some aspects, our construction is quite similar to the one used in~\cite{bodirsky2008non} for obtaining $\Pi^{\mathrm{P}}_{n}$-complete and ultimately \textrm{coNP}-intermediate $\omega$-categorical CSPs.
One notable difference is that~\cite[Theorem~2]{bodirsky2008non}  uses the Fra\"{i}ss\'{e} amalgamation method for constructing $\omega$-categorical CSP templates while our result uses~\cite[Corollary~17]{bodirsky2020datalog}, which itself goes back to~\cite[Theorem~4.27]{martin2012scope}. 
This observation raises the question whether one can replicate the construction of an intermediate $\omega$-categorical CSP for \textrm{NP} instead of \textrm{coNP}.
Below, we restate a question of Bodirsky~\cite[Question~49]{bodirsky2021complexity}. 

\medskip
\noindent\textbf{Question 3:} Assuming $\textrm{P}\neq \textrm{NP}$, does there exist an \textrm{NP}-intermediate $\omega$-categorical CSP?

\medskip Regarding the universal-algebraic upgrade~\cite[Theorem~1.8]{gillibert2022symmetries} to~\cite[Theorem~2]{bodirsky2008non}, it would be interesting to know if a similar result can be obtained for the $\Sigma^{\mathrm{P}}_{n}$-levels of \textrm{PH}. 
For $n=1$, this would confirm that the assumption of being a reduct of a finitely bounded homogeneous structure is essential in the algebraic dichotomy conjecture for infinite-domain CSPs~\cite{barto_pinsker_journal}.  

\medskip
\noindent\textbf{Question 4:} Does there exist an $\omega$-categorical structure $\struct{D}$ whose CSP is $\Sigma^{\mathrm{P}}_{n}$-complete for some $n\geq 1$ and such that the polymorphisms of $\struct{D}$ satisfy some non-trivial identities?

\bibliographystyle{plain}
\bibliography{references}
\end{document}